\documentclass[10pt]{article}

\usepackage[cmex10]{amsmath}
\usepackage{array}
\usepackage{stfloats}

\usepackage{amsfonts}%
\usepackage{amssymb}%
\usepackage{graphicx}
\usepackage[latin1]{inputenc}
\usepackage{hyperref}

\newtheorem{theorem}{Theorem}

\newtheorem{corollary}[theorem]{Corollary}

\newtheorem{definition}[theorem]{Definition}
\newtheorem{example}[theorem]{Example}

\newtheorem{lemma}[theorem]{Lemma}

\newtheorem{proposition}[theorem]{Proposition}
\newtheorem{remark}[theorem]{Remark}
\newtheorem{remarks}[theorem]{Remarks}

\newenvironment{proof}[1][Proof]{\textbf{#1.} }{\ \rule{0.5em}{0.5em}\\}

\newcommand{\N}{{\mathbb N}}
\newcommand{\Z}{{\mathbb Z}}
\newcommand{\F}{\mathbb F_q}
\newcommand{\B}{{\mathbb B}}
\renewcommand{\bf}[1]{\mathbf{#1}}
\newcommand{\Le}{\mathbb L}

\newcommand{\Q}{{\mathcal Q}}
\newcommand{\D}{{\mathcal D}}
\newcommand{\mad}{\mathrm{-mad}}

\newcommand{\tq}{\, \mid \,}
 
\newcommand{\supp}{{\rm supp}}   

\usepackage{color}

\begin{document}
\title{From ds-bounds for cyclic codes to true distance for abelian codes.}

 \author{ J.J. Bernal, M. Guerreiro\footnote{M. Guerreiro is with Departamento de Matem\'atica, Universidade Federal de Vi\c cosa, 36570-900 Vi\c cosa-MG, Brazil. Supported by CNPq-Brazil, Processo PDE 233497/2014-5.
E-mail: marines@ufv.br}~
and J. J. Sim\'on
\footnote{J. J. Bernal and J. J. Sim\'on are with the Departamento de Matem\'aticas, Universidad de Murcia, 30100 Murcia, Spain. Partially supported by MINECO project MTM2016-77445-P and Fundaci\'{o}n S\'{e}neca of Murcia.
E-mail: \{josejoaquin.bernal, jsimon\}@um.es}
\footnote{This work was partially presented at the ISITA 2016, Monterey, CA, USA.}
}

\maketitle


%
\begin{abstract}
In this paper we develop a technique to extend any bound for the minimum distance of cyclic codes constructed from its defining sets (ds-bounds) to abelian (or multivariate) codes through the notion of $\B$-apparent distance. We use this technique to improve the searching for new bounds for the minimum distance of abelian codes. We also study conditions for an abelian code to verify that its $\B$-apparent distance reaches its (true) minimum distance. Then we construct some tables of such codes as an application. 
\end{abstract}

\section{Introduction}

\noindent 
The study of abelian codes is an important topic in Coding Theory, having an extensive literature, because they have good algebraic properties that allow one to construct good codes with efficient encoding and decoding algorithms. More precisely, regarding decoding, the two most known general techniques are permutation decoding \cite{BSpermdec} and the so-called locator decoding~\cite{Blah} that uses the Berlekamp-Massey algorithm~\cite{Sakata} (see also~\cite{KZ}).

Even though the mentioned decoding methods require to know the minimum distance, or a bound for it, there are not much literature or studies on its computation and properties, or it does exist only for specific families of abelian codes (see~\cite{Blah}). Concerning BCH bound, in \cite{Camion}, P. Camion introduces an extension from cyclic to abelian codes which is computed through the apparent distance of such codes.  Since then there have appeared some papers improving the original computation and giving a notion of multivariate BCH bound and codes (see, for example,~\cite{BBCS2,  Evans}).

These advances lead us to some natural questions about the extension to the multivariate case of all generalizations and improvements of the BCH bound known for cyclic codes; specifically, those bounds on the minimum distance for cyclic codes which are defined from defining sets. 

There are dozens of papers on this topic regarding approaches from matrix methods (\cite{BettiSala, ZK1}) through split codes techniques (\cite{HTY,Jensen}) until arriving at the most classical generalizations based on computations over the 
defining set, as the Hartmann-Tzeng (HT)  bound~\cite{HT}, the Ross (R) bound~\cite{Roos} and the improvements by Van Lint and Wilson, as the shifting bound (SB) \cite{vLW}.

Having so many references on the subject, it seems very necessary to find a general method that allows one to extend any bound for the minimum distance of cyclic codes based on the defining set to the multivariate case. This is our first goal. We shall show a method to extend to the multivariate case any bound of the mentioned type via associating an apparent distance to such bound. 

The second target of our study is to improve the searching for new bounds for abelian codes. At this point, we must  honestly say that these searches may only have interest for codes whose minimum distances are not known (in fact, if one knows the minimum distance of a code one does not need a bound for it), so that, our examples consider codes of lengths necessarily large. 
In our opinion, long abelian codes are not so bad (see \cite{BJ}) in terms of performance.

As we work with long codes, certainly it seems impossible at the moment to compute their minimum distance, and so, it is natural to ask for conditions on them to ensure that a founded bound is in fact the minimum distance. This is the last goal of this paper. We found conditions for a bivariate abelian code to reach the mentioned equality and we write these conditions in terms of its defining set  from the notion of composed polynomial matrices (CP-matrices, for short). We comment the extension of this results to several variables. We illustrate with some examples (of large codes) how this technique works. 

In order to achieve our goals, we give in Section~\ref{sec3} a notion of defining set bound (ds-bound) for the minimum distance of cyclic codes. In Section~\ref{sec4}, we revisit the relation between the weight of codewords of abelian codes and the apparent distance of their discrete Fourier transforms. In Section~\ref{sec5}, we use this technique
to define the apparent distance of an abelian code with respect to a set of ds-bounds. In Section~\ref{sec6}, we adapt a known algorithm given in \cite{BBCS2} (of linear complexity by Remark \ref{complejidad}) to compute the $\B$-apparent distance of an abelian code. Finally, we study the abelian codes which verify the equality between its BCH bound and its minimum distance. For two variables, we find some  sufficient conditions that are easy to extend to several variables.
 
\section{Preliminaries} \label{sec2}

 Let $\F$  be a finite field with $q$ elements, with $q$ a power of a prime $p$, $r_i$ be positive integers, for all $i\in \{ 1,\ldots, s\}$, and $n=r_1\cdots r_s$.  We denote by $\Z_{r_i}$ the ring of integers modulo $r_i$ and we shall always write its elements as canonical representatives.
 
An \textbf{abelian code} of length $n$ is an ideal in the algebra   $\F(r_1,\ldots, r_s)=\F[X_1,\ldots, X_s]/\langle X_1^{r_1}-1,\ldots, X_s^{r_s}-1\rangle $ and throughout this work  we assume that this algebra is semisimple; that is, $\gcd (r_i,q)=1$, for all $i\in \{ 1,\dots,s\}$.  Abelian codes are also called multidimensional cyclic codes (see, for example, \cite{Imai}). 

The codewords are identified with polynomials $f(X_1,\dots,X_s)$ in which, for each monomial, the degree of the indeterminate $X_k$ belongs to $\Z_{r_k}$. We denote by $I$ the set $\Z_{r_1}\times\cdots\times \Z_{r_s}$ and we  write the elements $f \in  \F(r_1,\dots,r_s)$ as $f=f(X_1,\dots,X_s)=\sum a_\bf{i} \bf{X}^\bf{i}$, where $\bf{i}=(i_1,\dots, i_s)\in I$ and $\bf{X}^\bf{i}=X_1^{i_1}\cdots X_s^{i_s}$. Given a polynomial $f \in \F[X_1,\dots,X_s]$ we denote by $\overline{f}$ its image under the canonical projection onto $\F(r_1,\dots,r_s)$. 

For each $i\in \{ 1,\ldots, s\}$, we denote by $R_{r_i}$ (resp. $U_{r_i}$) the set of all $r_i$-th roots of unity (resp. all $r_i$-th primitive roots of unity) and define $R=\prod_{i=1}^s R_{r_i}$ ($U=\prod_{i=1}^s U_{r_i}$).

For $f=f(X_1,\dots,X_s) \in \F[X_1,\dots,X_s]$ and $\bar{\alpha}\in R$, we write $f(\bar{\alpha})=f(\alpha_1,\dots,\alpha_s)$. For $\bf{i}=(i_1,\ldots, i_s)\in I$, we write $\bar{\alpha}^{\bf{i}} = (\alpha_1^{i_1}, \dots ,\alpha_s^{i_s})$.

It is a known fact that every abelian code $C$ in $\F(r_1,\dots,r_s)$ is totally determined by its \textbf{root set} or \textbf{set of zeros}, namely
$$
Z(C)=\left\{\bar{\alpha}\in  R \tq f(\bar{\alpha})=0,\;\; \mbox{ for all }\; f\in C \right\}.
$$ 
The set of non zeros is denoted by $\overline{Z(C)}=R\setminus Z(C)$.  For a fixed $\bar{\alpha}\in U$, the code $C$ is  determined by its \textbf{defining set}, with respect to $\bar{\alpha}$, which is defined as 

$$
\D_{\bar{\alpha}}(C) = \left\{ \bf{i}\in I \tq f(\bar{\alpha}^{\bf{i}})=0, 
\text{ for all } f\in C\right\}.
$$

Given an element $a=(a_1,\dots,a_s)\in I$, we shall define its \textbf{$q$-orbit} modulo  $\left(r_1,\ldots,r_s \right)$ as $ Q(a)=\left\{\left(a_1\cdot q^{i} ,\dots, a_s\cdot q^{i}  \right) \in I \tq i\in \N\right\}$. In the case of a semisimple algebra, it is known that any defining set  $\D_{\bar{\alpha}}\left(C\right)$ is a disjoint union of $q$-orbits modulo $(r_1,\dots,r_s)$. Conversely, every union of $q$-orbits modulo $(r_1,\dots,r_s)$ determines an abelian code (an ideal) in $ \F(r_1,\dots,r_s)$  (see, for example, \cite{BBCS2} for details). We note that the notions of root set and defining set also apply to polynomials. Moreover, if $C$ is the ideal generated by the polynomial $f$ in $\F(r_1,\dots,r_n)$, then $\D_{\bar{\alpha}}\left(C\right)=\D_{\bar{\alpha}}\left(f\right)$.
 
We recall that the notion of defining set also applies to cyclic codes. For $s=1$ and  $r_1=n$, a $q$-orbit is called a \textbf{ $q$-cyclotomic coset} of a positive integer $b$ modulo $n$ and  it is the set $C_{q}(b) = \{b\cdot q^{i}\in \Z_{n} \tq i\in \N\}$.\\

Throughout this paper, we fix the notation $\Le|\F$ for an extension field containing $U_{r_i}$, for all  $i\in\{1,\dots,s\}$. The \textbf{discrete Fourier transform of a polynomial $f\in \F(r_1,\dots,r_s)$ with respect to $\bar{\alpha} \in U$} (also called Mattson-Solomon polynomial in \cite{Evans}) is the polynomial  
$\varphi_{\bar{\alpha},f}(\bf{X})=\sum_{\bf{j}\in I} f(\bar{\alpha}^{\bf{j}})\bf{X}^{\bf{j}}\in \Le(r_1,\dots,r_s).$
It is known that the discrete Fourier transform may be viewed as an isomorphism of algebras 
$\varphi_{\bar{\alpha}}:\Le(r_1,\dots,r_s)\longrightarrow (\Le^{|I|},\star),$
 where the multiplication ``$\star$'' in $\Le^{|I|}$ is defined coordinatewise. Thus, we may see $\varphi_{\bar{\alpha},f}$ as a vector in $\Le^{|I|}$ or as a polynomial in $\Le(r_1,\dots,r_s)$ (see \cite[Section 2.2]{Camion}).
 The inverse of the discrete Fourier transform is $\varphi^{-1}_{\bar{\alpha},g}(\bf{X})= \frac{1}{r_1r_2\cdots r_s} 
 \sum_{\bf{j}\in I} g(\bar{\alpha}^{-\bf{j}}) \bf{X}^\bf{j}$.

\section{Defining set bounds for cyclic codes} \label{sec3}

In this section we deal with cyclic codes; that is, $r_1=n$.  By $\mathcal{P}(\Z_n )$ we denote the power set of  $\Z_n $. We take an arbitrary $\alpha\in U_n$.

\begin{definition} \label{boundg}
A \textbf{defining set bound} (or \textbf{ds-bound}, for short) for the minimum distance  of cyclic codes  is a family of relations $\delta = \{\delta_n\}_{n\in \N} $ such that, for each $n\in \N $, $\delta_n \subseteq \mathcal{P}(\Z_n ) \times \N$ and it satisfies the following conditions:
\begin{enumerate}
\item If $C$ is a cyclic code in $\mathbb{F}(n)$ such that $ N \subseteq \D_{\alpha}(C)$, then $1\leq a\leq d(C)$, for all $(N, a) \in \delta_n$.
\item If $ N \subseteq M$ are subsets of  $\Z_n $ then
$(N, a) \in \delta_n$ implies $(M, a) \in \delta_n$.
\item For all $N \in  \mathcal{P}(\Z_n )$, $(N,1)\in \delta_n$. 
\end{enumerate}
\end{definition}

From now on, sometimes we write simply $\delta$ to denote a ds-bound or any of its elements independently on the length $n$ of the code. It will be clear from the context which one is being used.

\begin{remarks}\rm{
\textbf{(1)} For example, the BCH bound states that, for any cyclic code in $\mathbb{F}_q (n)$ which in its set of zeros has a string of $t -1$ consecutive powers of some $\alpha  \in U_n$, the minimum distance of the code is at least $t$ \cite[Theorem 7.8]{MWS}. 

Now, define $\delta \subset \mathcal{P}(\Z_n ) \times \N$ as follows: for any $a\geq 2$, $(N, a) \in \delta$ if and only if there exist $i_0, i_1,\ldots,i_{a-2}$  in $N$ which are consecutive integers modulo $n$. Then the  BCH bound says that $\delta$ is a ds-bound for any cyclic code (we only have to state Condition 3 as a convention; so that $(\emptyset,1)\in \delta_{BCH}$).

\textbf{(2)} It is easy to check that all extensions of the BCH bound, all new bounds  from the defining set of a cyclic code as in \cite{BettiSala,HT,Roos1,Roos,ZK1} and the new bounds and improvements arising from Corollary 1, Theorem 5 and results in Section 4 and Section 5 in \cite{vLW}, also verify Definition~\ref{boundg}.
}\end{remarks}

In general, for any bound for the minimum distance of a cyclic code, say $b$, we denote the corresponding ds-bound by $\delta_{b}$. In order to relate the idea of ds-bound with the Camion's apparent distance, which will be defined later, we consider the following family of maps.

\begin{definition}
Let $\delta$ be a ds-bound for the minimum distance of cyclic codes. The  \textbf{optimal ds-bound associated to $\delta$} is the family $\overline{\delta}= \{\overline{\delta}_n\}_{n\in \N} $ of maps $\overline{\delta}_n :  \mathcal{P}(\Z_n ) \longrightarrow  \N$ defined as 
$\overline{\delta}_n(N) = \max \{ b\in \N \,|\, (N,b) \in \delta_n \}$.
\end{definition}

The following result is immediate.
\begin{lemma} \label{lemma1} 
Let $\delta$ be a ds-bound for the minimum distance of cyclic codes.   Then, for each $n\in \Z$:
\begin{enumerate}
\item If $C$ is a cyclic code in   $\mathbb{F}(n)$ such that $ N \subseteq \D_{\alpha}(C)$, then $1\leq \overline{\delta}_n (N)\leq d(C)$.

\item If $ N \subseteq M\subseteq \Z_n $, then
$\overline{\delta}_n (N) \leq \overline{\delta}_n (M)$.
\end{enumerate}
\end{lemma}
As we noted above, we may omit the index of the map $\bar{\delta}_n$, because it will be clear from the context 
for which value it is being taken.

\section{Apparent distance of matrices} \label{sec4}

We begin this section recalling the notion and notation of a hypermatrix that will be used hereby, as it is described in \cite{BBCS2}. For any $\bf{i}\in I= \Z_{r_1}\times\cdots\times \Z_{r_s}$, we write its $k$-th coordinate as $\bf{i}(k)$. A \textbf{hypermatrix with entries in a field $S$ indexed by $I$ (or an $I$-hypermatrix over $S$)} is an $s$-dimensional $I$-array, denoted by $M=\left(a_{\bf{i}}\right)_{\bf{i}\in I}$, with $a_{\bf{i}}\in S$ \cite{Yamada}. The set of indices, the dimension and the ground field will be omitted if they are clear from the context. For $s=2$, $M$ is a matrix and when $s=1$, $M$ is a vector. We write $M=0$ when all its entries are $0$ and   $M\neq 0$, otherwise.  As usual, a \textbf{hypercolumn} is defined as $H_M(j,k)=\left\{ a_{\bf{i}}\in M\tq \bf{i}(j)=k\right\}$, with $1\leq j\leq s$ and $0\leq k<r_j$, where $a_{\bf{i}}\in M$ means that $a_{\bf{i}}$ is an entry of $M$. A hypercolumn can be seen as an $(s-1)$-dimensional hypermatrix. In the case $s=2$, we refer to hypercolumns as rows or columns and, when $s=1$, we say entries.

 For any $I$-hypermatrix $M$ with entries in a field, we define the support of $M$ as the set $\supp(M)=\left\{ \bf{i}\in I \tq a_{\bf{i}}\neq 0\right\}$. Its complement with respect to $I$ will be denoted by $\D(M)$. When $\D(M)$ (or $\supp(M)$) is an union of $q$-orbits we say that $M$ is a \textbf{$q$-orbits hypermatrix}.
Let $D\subseteq I$. The \textbf{hypermatrix afforded by $D$} is defined as $M=\left(a_{\bf{i}}\right)_{\bf{i}\in I}$, where $a_{\bf{i}}=1$ if $\bf{i}\not\in D$ and $a_{\bf{i}}=0$ otherwise; it will be denoted by $M=M(D)$. Note that if $D$ is union of $q$-orbits then $M(D)$ is a $q$-orbits hypermatrix. To define and compute the apparent distance of an abelian code we will use the hypermatrix afforded by its defining set, with respect to $\bar{\alpha} \in U$.

 We define a partial ordering on the set $\left\{M(D)\tq D\text{ is union of $q$-orbits in $I$}\right\}$  as follows:
\begin{equation}\label{matrixordering}
M(D)\leq M(D') \Leftrightarrow \supp\left(M(D)\right)\subseteq \supp\left(M(D')\right).
\end{equation}
Clearly, this condition is equivalent to $D'\subseteq D$. 

We begin with the apparent distance of a vector in $\Le^{n}$.

\begin{definition} 
Let $\delta$ be a ds-bound for the minimum distance of cyclic codes and $v\in \mathbb{L}^n$ a vector. The \textbf{apparent distance of $v$ with respect to $\delta$} (or \textbf{$\delta$-apparent distance of $v$}, for short), denoted by $\delta^*(v)$, is defined as 
\begin{enumerate}
\item If $v=0$, then $\delta^*(v)=0$.
\item If $v\neq 0$, then $\delta^*(v)= \overline{\delta}(\mathbb{Z}_n\setminus \supp(v))$. 
\end{enumerate}
\end{definition}
From now on we denote by $\B$ a set of ds-bounds which are used to proceed a  computation of the apparent distances of matrices, hypermatrices or abelian codes. 

\begin{definition} \label{appdistbound}
Let $v\in \mathbb{L}^n$. The \textbf{apparent distance of $v$ with respect to $\B$}  denoted by $\Delta_\B(v)$, is: 
\begin{enumerate}
\item If $v=0$, then $\Delta_\B(v)=0$.
\item If $v\neq 0$, then $\Delta_\B(v)= \max\{\delta^*(v)\mid \delta\in\B\}$. 
\end{enumerate}
\end{definition}

\begin{remarks}\label{propiedades dist apar en 1 var}
 The following properties arise straightforward from the definition above, for any $v\in \mathbb{L}^n$.
 \begin{enumerate}
  \item If $v\neq 0$ then $\Delta_{\B}(v)\geq 1$.
  \item If $\supp (v)\subseteq \supp(w)$ then $\Delta_{\B}(v)\geq \Delta_{\B}(w)$. \\
 \end{enumerate}
\end{remarks}

\begin{proposition}\label{prop1}
Let $f\in \Le(n)$ and $v$ be the vector of its coefficients. Fix any $\alpha\in U_n$. Then $ \Delta_{\B}(v) \leq \omega(\varphi_{\alpha,f}^{-1})=|\overline{Z(f)}|$.
\end{proposition}
\begin{proof}
  Set $N=\mathbb{Z}_n\setminus \supp(v)$ and let  $C$ be the abelian code generated by $\varphi_{\alpha,f}^{-1} $ in $\mathbb{L}(n)$. Then $d(C)\leq \omega(\varphi_{\alpha,f}^{-1})$.  By properties of the discrete Fourier transform, we have $N=\D_{\alpha}(\varphi_{\alpha,f}^{-1} )=\D_{\alpha}(C)$ hence, by Lemma~\ref{lemma1} and the definition of apparent distance, $\Delta_{\B}(v)\leq d(C)$. This gives the desired inequality. The last equality is obviuos.
\end{proof}

The notion of apparent distance appeared for the first time in \cite{Camion} and originally it was  defined  for polynomials. Its computation reflects a bound of the nonzeros (in the sense given in the preliminaries) of a \textit{multivariate} polynomial. The aim of the apparent distance was to extend the notion of BCH bound, from cyclic to abelian codes as we will comment in the following paragraphs. The first algorithm for its computation was made in terms of coefficients of polynomials. Later, in \cite{Evans}, R. E. Sabin gave an algorithm in terms of matrices. The notion of strong apparent distance, that appeared in \cite{BBCS2}, is a slight but powerful modification of the original one, defined for multivariate polynomials and hypermatrices, and it is the predecessor of the current apparent distance defined with respect to a list of ds-bounds.

\begin{remark}\label{d* y DeltaBCH}\rm{
To identify notations from previous works with the ones used here, given a polynomial $f\in \F(r_1,\dots,r_s)$, if we denote by $M(f)$ its hypermatrix of coefficients (in the obvious sense), then the strong apparent distance of $f$ in \cite{BBCS2} is $sd^*(f)=\Delta_{\delta_{BCH}}(M(f))$; that is, $\B=\{\delta_{BCH}\}$, together with the convention that $\delta_{BCH}(\emptyset)=1$.
}\end{remark}

Now let us show how the notion of apparent distance for abelian codes works as the BCH bound for cyclic codes. All results in the following corollary are proved in \cite{Camion} and \cite{Evans}.

\begin{corollary} \label{calculo dis aparent codigo ciclico}
 Let $C$ be a cyclic code in $\F(n)$ and $\alpha\in U_n$. Then
 \begin{enumerate} 
 \item If $g,e\in C$ are the generating polynomial and the idempotent generator of $C$, respectively, then $\Delta_{\B}\left[M\left(\varphi_{\alpha,g}\right)\right]=\Delta_{\B}\left[M\left(\varphi_{\alpha,e}\right)\right]\leq \Delta_{\B}\left[M\left(\varphi_{\alpha,c}\right)\right]$, for all $c\in C$.
  \item If $c\in C$ is a codeword with $\varphi_{\alpha,c}=f\in \Le(n)$, then $\omega(c)\geq \Delta_{\B}(M(f))$ and consequently
  \item $\Delta_{\B}\left[M\left(\varphi_{\alpha,g}\right)\right]=\Delta_{\B}\left[M\left(\varphi_{\alpha,e}\right)\right]=\min\left\{\Delta_{\B}\left[M\left(\varphi_{\alpha,c}\right)\right]\tq c\in C\right\}\leq d(C)$.
 \end{enumerate} 
 
 The number on the left of the last inequality is known as the \textbf{apparent distance of the cyclic code $C$ with respect to the set $\B$ and $\alpha \in U_n$} or the $\B$-apparent distance of $C$ with respect to $\alpha \in U_n$.
\end{corollary}
\begin{proof}
\textit{(1)} comes from the fact that, for all $c\in C$, we have  $\supp\left(M\left(\varphi_{\alpha,c}\right)\right)\subseteq \supp\left(M\left(\varphi_{\alpha,g}\right)\right)=\supp\left(M\left(\varphi_{\alpha,e}\right)\right)$, together with Remark~\ref{propiedades dist apar en 1 var}. \textit{(2)} is Proposition~\ref{prop1}. \textit{(3)} is immediate from \textit{(1)} and \textit{(2)}.
\end{proof}

Now we shall define the apparent distance of matrices and hypermatrices with respect to a set $\B$ of ds-bounds.

\begin{definition} \label{Bappdistmatrix}
Let $M$ be an $s$-dimensional $I$-hypermatrix over a field $\mathbb{L}$. The \textbf{apparent distance of $M$ with respect to $\B$}, denoted by $\Delta_{\B}(M)$,  is defined as follows:

\begin{enumerate}
 \item $\Delta_{\B}(0)=0$ and, for $s=1$,  Definition~\ref{appdistbound} applies.
\item For $s=2$ and a nonzero matrix $M$, note that $H_M(1,i)$ is the $i$-th row  and $H_M(2,j)$ is the $j$-th column of $M$. Define the \textbf{row support of $M$} as $\supp_1(M)= \{ i\in \{0, \ldots r_1-1\}  \,|\, H_M(1,i)\neq 0 \}$ and the \textbf{column support of $M$} as $\supp_2(M)= \{ k\in \{0, \ldots r_2-1\}  \,|\, H_M(2,k)\neq 0 \} $.

Then put
\begin{eqnarray*}
 \omega_{1}(M)&=& \max\{\overline{\delta}(\mathbb{Z}_{r_1}\setminus \supp_1(M))\mid \delta\in\B\},\\
 \epsilon_{1}(M)&=&\max \{ \Delta_\B(H_M(1,j)) \tq j\in \supp_1(M)\}
\end{eqnarray*}
and set $\Delta_{1}(M)=\omega_{1}(M)\cdot \epsilon_{1}(M)$.

Analogously, we compute the apparent distance $\Delta_2(M)$ for the other  variable and finally we define the \textbf{apparent distance of $M$ with respect to $\B$} by
\[
\Delta_{\B}(M) = \max \{\Delta_1(M),\Delta_{2}(M)\}.
\]

\item For $s>2$, proceed as follows: suppose that one knows how  to compute the apparent distance 
$\Delta_{\B}(N)$, for all non zero  hypermatrices $N$ of dimension $s-1$. Then first compute the ``hypermatrix support'' of $M\neq 0$ with respect to the $j$-th hypercolumn, that is, $$\supp_j(M) = \{ i\in 
\{0, \ldots r_j-1\} \,|\, H_M(j,i)\neq 0\}.$$ 

Now put
\begin{eqnarray*}
 \omega_{j}(M)&=& \max\{\overline{\delta}(\mathbb{Z}_{r_j}\setminus \supp_j(M))\mid \delta\in \B\}, \\
 \epsilon_{j}(M)&=&\max \{ \Delta_{\B}(H_M(j,k)) \,|\,k\in \supp_j(M)\}
\end{eqnarray*}

and set $\Delta_j(M)=\omega_{j}(M)\cdot \epsilon_{j}(M)$. 

Finally, define the \textbf{apparent distance of $M$ with respect to $\B$} (or the $\B$-apparent distance) as:
$$\Delta_{\B}(M)= \max \left\{ \Delta_j(M)\tq j\in\{1,\dots,s\}\right\}.$$
 \end{enumerate}
\end{definition}

As we have already commented in Remark~\ref{d* y DeltaBCH}, by taking $\B=\{\delta_{BCH}\}$, $\Delta_{\B}(M)$ is the strong apparent distance in \cite{BBCS2}.

Now, as in the case of cyclic codes, we relate the apparent distance to the weight of codewords. For each multivariate polynomial $f=\sum_{\mathbf{i}\in I} a_{\mathbf{i}}\mathbf{X}^{\mathbf{i}}$,  consider the hypermatrix  of the coefficients of $f$, denoted by $M(f) = (a_{\mathbf{i}})_{\mathbf{i}\in I}$. For any $j \in \{ 0, \ldots, s\}$, if we write $f=\sum_{k=0}^{r_j-1}f_{j,k}X_j^k$, where $f_{j,k}=f_{j,k}(\mathbf{X}_j)$ and $\mathbf{X}_j=X_1\cdots X_{j-1}\cdot X_{j+1} \cdots X_s$, then $M(f_{j,k})= H_M(j,k)$. This means that ``fixed" the variable $X_j$ in $f$, for each power $k$ of $X_j$, the coefficient $f_{j,k}$ is a polynomial in $\mathbf{X}_j$, and $H_M(j,k)$ is the hypermatrix obtained from its coefficients. Now we extend Proposition~\ref{prop1} to several variables.

\begin{theorem}\label{boundDFTnvar} Let $f\in \mathbb{L}(r_1,\ldots, r_s)$ and $M=M(f)$ be the hypermatrix of its coefficients. Fix any $\bar{\alpha}\in U$.  Then  $\Delta_{\B}(M(f)) \leq \omega\left(\varphi^{-1}_{\overline{\alpha},f}\right)= |\overline{Z(f)} | $.
 \end{theorem}
 
 \begin{proof}
 For $f=0$, the result is obvious. Consider $f\neq 0$. The case $n=1$ is Proposition~\ref{prop1}. We prove the theorem for matrices; that is, for $s=2$. The general case follows directly by induction.
  
  Set $M=M(f)\neq 0$, $\bar{\alpha}=(\alpha_1,\alpha_2)$ and write $f=\sum_{k=0}^{r_2-1}f_{2,k}X_2^k$. Then $H_M(2,k)$ is the vector of coefficients of $f_{2,k}$. Clearly, $\supp_2(M)=\{k\in\{0,\dots,r_2-1\}\tq f_{2,k}\neq 0\}$ and, for any $k\in \supp_2(M)$, we have $\Delta_{\B}\left(H_M(2,k)\right)\geq 1$. Then $\omega\left(\varphi^{-1}_{\alpha_1,f_{2,k}}\right)\geq 1$, by Proposition~\ref{prop1}.
  
 Now, for each fixed $k\in\Z_{r_2}$, by the definition of discrete Fourier transform, 
 $ |\overline{Z(f_{2,k})} | = \omega(\varphi_{\alpha_1,f_{2,k}}^{-1})$, hence, if $k\in \supp_2(M)$ there exists $t\in \Z_{r_1}$ (at least one) such that $\alpha_1^t$ is a non zero of $f_{2,k}$ .
  
 Set $ g(X_2)=f(\alpha_1^t,X_2)=\sum_{k=0}^{r_2-1}f_{2,k}(\alpha_1^t)X_2^k$ and note that it is a non zero polynomial. Let $v(g)$ be the vector of coefficients of $g(X_2)$. Then $\supp\left(v(g)\right)\subseteq \supp_2(M)$ and so $\Delta_{\B}(v(g))\geq \max\{\bar{\delta}\left(\Z_{r_2}\setminus \supp_2(M)\right)\tq \delta\in \B\}=\omega_2(M)$.
  
 As $|\overline{Z(g)}|=\omega(\varphi_{\alpha_2,g}^{-1})\geq \Delta_{\B}(v(g))$ then, for any $k\in\Z_{r_2}$, 
 $$|\overline{Z(f)}|\geq |\overline{Z(g)}|\cdot |\overline{Z(f_{2,k})}| \geq \Delta_{\B}(v(g))\cdot \omega(\varphi_{\alpha_1,f_{2,k}}^{-1})  \geq \omega_2(M)\cdot\Delta_{\B}\left(H_M(2,k)\right).$$ Finally, as in the univariate case, it is clear that $ |\overline{Z(f)} | = \omega(\varphi_{\bar{\alpha},f}^{-1})$. The extension to more variables is clear and this completes the proof.
 \end{proof}

\begin{example}\label{Ex 4por24} \rm{ 
Set $n=96=4\times 24$ and $q=5$. Fix $\alpha_1\in U_{4}$ and $\alpha_2\in U_{24}$ and consider the $5$-orbits matrix $M$ afforded by $D=Q(0,0)\cup Q(0,1)\cup Q(0,2)\cup Q(0,3) \cup  Q(0,6)\cup  Q(0,7) \cup Q(0,9) \cup  Q(1,1) \cup  Q(1,2) \cup  Q(1,3) \cup  Q(2,1) \cup  Q(2,2) \cup  Q(3,6).$ Choose $\B=\{\delta_{BS},\delta_{BCH}\}$, where  
$\delta_{BS}$ is the Betti-Sala bound in \cite{BettiSala}. 

One may check that $\supp_1(M)=\{0,1,2,3\}$ and then $\omega_1(M)=1$. On the other hand, $\supp_2(M)=\Z_{24}$ so $\omega_2(M)=1$. Now $\Delta_{\B}(H_M(1,0))=8$, by using $\delta_{BS}$ (see \cite[Example 4.2]{BettiSala}) which is the maximum of the values of the two bounds considered, hence $\epsilon_1(M)=8$. It is clear that $\epsilon_2(M)=4$, so that $\Delta_1(M)=8$  and $\Delta_2(M)=4$. Hence $\Delta_{\{\delta_{BS},\delta_{BCH}\}}(M)=8$.
} \end{example}

The computation of the apparent distance in several variables is a natural extension of that of one variable, and, moreover, the relationship between apparent distance and weight of a codeword is essentially the same in any case. However, condition (2) of Remark~\ref{propiedades dist apar en 1 var}\textit{(2)} does not necessarily hold in two or more variables and so we cannot extend directly the results of Corollary~\ref{calculo dis aparent codigo ciclico}. Let us show the situation in the following example.

\begin{example}\rm{
 Let $M$ be the $2$-orbits matrix of order $5\times 7$  such that $\supp(M)= Q(0,0)\cup Q(1,0)\cup Q(1,3)$ and $N$ be the $2$-orbits matrix such that $\supp(N)= Q(1,0)\cup Q(1,3)$. Then $N<M$, however, one may check that $\Delta_{\delta_{BCH}}(N)=6$ and $\Delta_{\delta_{BCH}}(M)=7$. So Remark~\ref{propiedades dist apar en 1 var}\textit{(2)} does not hold in this case.
}\end{example}

As we will see  in the next section, the notion of $\B$-apparent distance of an abelian code with respect to some roots of unity will be a natural extension of that of cyclic codes, while its computation  will be an interesting problem to solve.

\section{The $\B$-apparent distance of an abelian code} \label{sec5}
The following definition changes a little the usual way to present the notions of apparent distance from \cite{Camion} and the strong apparent distance from \cite{BBCS2} (see also \cite{Evans}). We recall that $\B$ denotes a set of ds-bounds, which are used to proceed a concrete computation of the apparent distances.

\begin{definition} \label{appdistcodenvar}
Let $C$ be an abelian code in $\F(r_1,\dots,r_s)$.

\textit{1)} The \textbf{apparent distance of $C$ with respect to  $\bar{\alpha}\in U$ and $\B$} (or the $(\B,\bar{\alpha})$-apparent distance) is
$$  
\Delta_{\B,\overline{\alpha}}(C)= \min \{ \Delta_{\B} (M\left(\varphi_{\bar{\alpha},c})\right) \,|\, c\in C\}.
$$

\textit{2)} The \textbf{apparent distance of $C$ with respect to  $\mathbb{B}$} is
 $$\Delta_{\mathbb{B}}(C)= \max \{ \Delta_{\B,\bar{\alpha}}(C) \,|\, \bar{\alpha}\in U\}.$$
\end{definition}

The following result is a consequence of Theorem~\ref{boundDFTnvar}.\\

\begin{corollary} \label{corTheo}
For any abelian code $C$ in $\F(r_1,\dots,r_s)$ and any $\B$ as above, $\Delta_{\mathbb{B}}(C) \leq d(C)$.
\end{corollary}
\begin{proof}
 Let $g\in C$ such that $\omega(g)=d(C)$. By Theorem~\ref{boundDFTnvar}, $\Delta_{\B}\left(M\left(\varphi_{\bar{\alpha},g}\right)\right)\leq \omega(g)$, for any $\bar{\alpha}\in U$. From this, the result follows directly.
\end{proof}\\

It is certain that to compute the apparent distance for each element of a code in order to obtain its apparent distance can be as hard work as to compute the minimum distance of such a code. Thus, to improve the efficiency of the computation the following result tells us that we may restrict our attention to the idempotents of the code. It also allows us to reformulate Definition \ref{appdistcodenvar} as it is presented in the previously mentioned papers.

 \begin{proposition} \label{prop2}
Let $C$ be an abelian code in $\F(r_1,\dots,r_s)$. The apparent distance of $C$ with respect to  $\bar{\alpha}\in U$ and $\B$ verifies the equality
$$ \Delta_{\B, \bar{\alpha}}(C) = \min \{ \Delta_{\B} \left(M(\varphi_{\bar{\alpha},e})\right) \,|\, e^2=e\in C\}.$$
\end{proposition}
\begin{proof}
 Consider any $c\in C$. Since $\F(r_1,\dots,r_s)$ is semisimple, then there exists an idempotent $e\in C$ such that the ideals generated by $c$ and $e$ in $\F(r_1,\dots,r_s)$ coincide; that is, $\langle c\rangle =\langle e \rangle$, and so $\D_{\bar{\alpha}}(c)=\D_{\bar{\alpha}}(e)$. This means that $\supp\left(M(\varphi_{\bar{\alpha},c})\right) =\supp\left(M(\varphi_{\bar{\alpha},e})\right)$. Note that the computation of the apparent distance is based on the fact that the entries (of the matrices) are zero or not; that is, once an entry is non zero, its specific value is irrelevant. From this fact, it is easy to see that $ \Delta_{\B} \left(M(\varphi_{\bar{\alpha},c})\right)= \Delta_{\B} \left(M(\varphi_{\bar{\alpha},e})\right)$ and so we get the desired equality.
\end{proof}

Let $e\in\F(r_1,\dots,r_s)$ be an idempotent and $E$ be the ideal generated by $e$. Then $\varphi_{\bar{\alpha}, e}\star\varphi_{\bar{\alpha}, e}=\varphi_{\bar{\alpha}, e}$, for any $\bar{\alpha} \in U$ and thus, if $\varphi_{\bar{\alpha}, e}=\sum_{\bf{i}\in I} a_{\bf{i}}X^{\bf{i}}$, we have $a_{\bf{i}}\in \{1,0\}\subseteq \F$ and $a_{\bf{i}}=0$ if and only if $\bf{i}\in \D_{\bar{\alpha}}(E)$. Hence $M(\varphi_{\bar{\alpha},e})=M(\D_{\bar{\alpha}}(E))$. Conversely, let $M$ be a  hypermatrix afforded by a set $D$ which is a union of $q$-orbits. We know that $D$ determines a unique ideal $C$ in $\F(r_1,\dots,r_s)$ such that $\D_{\bar{\alpha}}(C)=D$. Let $e\in C$ be its generating idempotent. Clearly,  $M(\varphi_{\bar{\alpha}, e})=M(D)$.

Now let $C$ be an abelian code, $\bar{\alpha} \in U$ and let $M$ be the hypermatrix afforded by $\D_{\bar{\alpha}}(C)$. For any $q$-orbits hypermatrix $P\leq M$ [see the ordering (\ref{matrixordering})] there exists a unique idempotent $e'\in C$ such that $P=M(\varphi_{\bar{\alpha}, e'})$ and, for any codeword $f\in C$, there is a unique idempotent $e(f)$ such that $\Delta_{\B}\left(M(\varphi_{\bar{\alpha},f})\right)= \Delta_{\B} \left(M(\varphi_{\bar{\alpha},e(f)})\right)$. Therefore,
\begin{eqnarray*}
  \min\{\Delta_{\B}(P)\tq 0\neq P\leq M\}= \\ \min\{\Delta_{\B}(M(\varphi_{\bar{\alpha}, e}))\tq
  0\neq e^2=e\in C\}= \Delta_{\B,\bar{\alpha}}(C).
\end{eqnarray*}
This fact drives us to give the following definition.
\begin{definition}
 For a $q$-orbits hypermatrix $M$, its \textbf{minimum $\mathbb{B}$-apparent distance} is
	\[\mathbb{B}\mad(M)=\min\{\Delta_{\B}(P)\tq 0\neq P\leq M\}.\]
\end{definition}

Finally, in the next theorem we set the relationship between the apparent distance of an abelian code and the minimum apparent distance of hypermatrices.

\begin{theorem}
Let $C$ be an abelian code in $\F(r_1,\dots,r_s)$ and let $e$ be its generating idempotent. For any $\bar{\alpha}\in U$, we have $\Delta_{\B,\bar{\alpha}}(C)= \mathbb{B}\mad\left(M(\varphi_{\bar{\alpha},e})\right)$. Therefore, $\Delta_{\mathbb{B}} (C)= \max\{\mathbb{B}\mad\left(M(\varphi_{\bar{\alpha},e})\right)\tq \bar{\alpha}\in U\}$.
\end{theorem}
\begin{proof}
It follows  directly from the preceding  paragraphs.
\end{proof}

\section{Computing minimum apparent distance} \label{sec6}

In \cite{BBCS2} it is presented an algorithm to find, for any abelian code, a list of matrices (or hypermatrices in case of more than 2 variables) representing some of its idempotents  whose apparent distances based on the BCH bound (called the strong apparent distance) go decreasing until the minimum value is reached. It is a kind of ``suitable idempotents chase through hypermatrices'' \cite[p. 2]{BBCS2}. This algorithm is based on certain manipulations of the ($q$-orbits) hypermatrix afforded by the defining set of the abelian code. It is not so hard to see that it is possible to obtain an analogous algorithm in our case.

We reproduce here the result and the algorithm in the case of two variables under our notation. Then we will use the mentioned algorithm to improve the searching for new bounds for abelian codes.

\begin{definition}\label{indices involucradas}
 With the notation of the previous sections, let $D$ be a union of $q$-orbits and $M=M(D)$ the hypermatrix afforded by $D$. We say that $H_M(j,k)$ is an \textbf{involved hypercolumn (row or column for two variables) in the computation of $\Delta_{\B}(M)$}, if $\Delta_{\B}(H_M(j,k))=\epsilon_{j}(M)$ and $\Delta_{j}(M)=\Delta_{\B}(M)$.
\end{definition}

 We denote the set of indices of involved hypercolumns by $I_p(M)$.
Note that the involved hypercolumns are those which contribute in the computation of the $\B$-apparent distance.

The next result shows a sufficient condition to get at once the minimum $\mathbb{B}$-apparent distance of a hypermatrix.

\begin{proposition}\label{matrizdamvarias}
With the notation as above, let $D$ be a union of $q$-orbits and $M=M(D)$ the hypermatrix afforded by $D$. If $\Delta_{\B}(H_M(j,k))=1$, for some $(j,k)\in I_p(M)$, then $\mathbb{B}\mad(M)=\Delta_{\B}(M)$. 
\end{proposition}
\begin{proof}
 It is a modification of that in \cite[Proposition 23]{BBCS2} having in mind the use of different ds-bounds.
\end{proof}

\begin{theorem} \label{teodam2}
Let $\Q$ be the set of all $q$-orbits modulo $(r_1,r_2)$, $\mu\in \{1,\dots, |\Q|-1\}$ and $\left\{Q_j\right\}_{j=1}^{\mu} $ a subset of $\Q$. Set $D=\cup_{j=1}^\mu Q_j$ and $M=M(D)$. Then there exist two sequences: the first one is formed by nonzero $q$-orbits matrices, $M=M_0 > \dots > M_l\neq 0$ and the second one is formed by positive integers
$m_0 \geq  \dots \geq m_l, $ with $ l\leq \mu$ and $m_i\leq \Delta_{\B}(M_i)$, verifying the following property: 

If $P$ is a $q$-orbits matrix such that $0 \neq P \leq M$, then  $\Delta_{\B}(P) \geq m_l$ and if $\Delta_{\B}(P) < m_{i-1}$ then $P \leq M_i$, where $0 < i\leq l$.

 Moreover,  if $l' \in \{0, \dots , l\}$ is the first element satisfying $m_{l'}=m_l$ then 
$\Delta_{\B}(M_{l'})=\mathbb{B}\mad(M)$. 
\end{theorem}
\begin{proof}
 It follows the same lines of that in \cite[Proposition 25]{BBCS2} having in mind the use of different ds-bounds.
\end{proof}

\noindent\textbf{Algorithm for matrices.}   

Set $I=\Z_{r_1}\times\Z_{r_2}$. Consider the $q$-orbits matrix $M=\left(m_{ij}\right)_{(i,j)\in I}$ and a set $\B$ of ds-bounds.
\begin{itemize}
\item[]Step 1. Compute the apparent distance of $M$ with respect to $\B$ and set $m_0= \Delta_{\B}(M)$.

\item[]Step 2. 
  \begin{itemize}
      \item[a)]  If there exists $(j,k)\in I_p(M)$ (see notation below Definition~\ref{indices involucradas}) such that $\Delta_{\B}(H_M(j,k))=1$, then we finish giving the sequences $M=M_0$ and $m_0=\Delta_{\B}(M)$ (because of Proposition~\ref{matrizdamvarias}).
      
      \item[b)]  If $\Delta_{\B}(H_M(j,k)) \neq 1$, for all $(k,b)\in I_p(M)$, we set $$S=\bigcup_{ (k,b)\in I_p(M)} \supp (H_M(k,b))$$ and construct the $q$-orbits matrix $M_1=\left(a_{ij}\right)_{(i,j)\in I}$ such that 
      \[a_{ij}=
      \begin{cases} 
      0 & \text{if}\; (i,j)\in \cup \{Q(k,b)\tq (k,b)\in S\} \\
      m_{ij} & \text{otherwise.}  
    \end{cases}\]
  \end{itemize}
In other words, $M_1<M$ is the ($q$-orbits) matrix with maximum support such that the involved rows  and columns of $M$ 
are replaced by zero. One may prove that if $0\neq P <M$ and $\Delta_{\B}(P)< m_0$ then $P \leq M_1$.

\item[]Step 3.
  \begin{itemize}
      \item [a)] If $M_1=0$, then we finish giving the sequences $M=M_0$ and $m_0= \Delta_{\B}(M)$.
      \item [b)] If $M_1 \neq 0$, we set $m_1= \min \{m_0, \Delta_{\B}(M_1)\}$, and we get the sequences $M=M_0>M_1$ and $m_0\geq m_1$. Then, we go back to Step 1 with  $M_1$ in the place of $M$ and $m_1$ in the place  of $m_0$. $\blacksquare$\\
  \end{itemize}
 \end{itemize}
\begin{remark}\label{complejidad}
If the $q$-orbits matrix has $\mu$ $q$-orbits, the algorithm has at most $\mu$ steps. $\blacksquare$
\end{remark}


\begin{example}\label{sec7}{\rm
We take the setting of Example~\ref{Ex 4por24} and consider the abelian code $C$ with $\D(C)=D$. In this case, the matrix $M=M(\D(C))$ is the same as that in the mentioned example. Choose again $\B=\{\delta_{BCH},\delta_{BS}\}$. This code has $\dim_{\mathbb{F}_5}(C)=73$ and $\Delta_{\B}(C)=8$. 

 As $I_p(M)=\{(1,0)\}$, the matrix $M_1$ has the first row all zero and the others equal to the ones of $M$; that is $\D(M_1)=\left(\cup_{i\in\Z_{24}}Q(0,i)\right)\cup Q(1,1) \cup  Q(1,2) \cup  Q(1,3) \cup  Q(2,1) \cup  Q(2,2) \cup  Q(3,6).$
 
 Now $\supp_1(M_1)=\{1,2,3\}$, $\omega_1(M_1)=2$ and $\epsilon_1(M_1)=4$ and we also have $\supp_2(M_1)=\{0,\dots,23\}$, $\omega_2(M_1)=1$ and  $\epsilon_2(M_1)=4$. Hence $\Delta_{\mathbb{B}}(M_1)=8$. Here $I_p(M_1)=\{(1,1)\}$ and we get the matrix $M_2$ having its first and second rows all zero and the others equal to the ones of $M$ and $M_1$; that is, $\D(M_2)=\left(\cup_{i\in\Z_{24}}Q(0,i)\right)\cup \left( \cup_{i\in\Z_{24}}Q(1,i)\right) \cup Q(2,1) \cup  Q(2,2) \cup  Q(3,6).$

Here $\supp_1(M_2)=\{2,3\}$, $\omega_1(M_2)=3$ and $\epsilon_1(M_2)=3$ and we also have $\supp_2(M_2)=\{0,\dots,23\}$, $\omega_2(M_2)=1$ and  $\epsilon_2(M_2)=4$. Hence $\Delta_{\mathbb{B}}(M_2)=9$ and $I_p(M_2)=\{(1,2)\}$ and we get the matrix $M_3$ having its first, second and third rows all zero and the others equal to the ones of $M$, $M_1$ and $M_2$; that is, $\D(M_2)=\left(\cup_{i\in\Z_{24}}Q(0,i)\right)\cup\left( \cup_{i\in\Z_{24}}Q(1,i)\right)\cup \left(\cup_{i\in\Z_{24}}Q(2,i)\right) \cup  Q(3,6).$

Finally, $\supp_1(M_3)=\{3\}$, $\omega_1(M_2)=4$ and $\epsilon_1(M_2)=2$; so we also have $\supp_2(M_2)=\Z_{24}\setminus \{6\}$, $\omega_2(M_3)=2$ and  $\epsilon_2(M_3)=4$. Hence $\Delta_{\mathbb{B}}(M_3)=8$ and $I_p(M_3)=\{(1,3)\}$ and we get $M_4=0$. Therefore, $\Delta_{\{\delta_{BCH},\delta_{BS}\}}(M)=8$.\\

The closest code to $C$ we know is a $(105,51,7)$ binary cyclic code in \cite[Table II]{HTY}. The known bounds for linear codes with the same length and dimension are between 10 and 15.
}\end{example}

The reader may find some tables with examples of this kind in \cite{BGSISITA}.

\section{True minimum distance in abelian codes} \label{sec8}

In this section we study the problem of find abelian codes such that its apparent distance or its multivariate BCH bound reaches its minimum distance. We keep all the notation from the preceding sections. In~\cite{BBS1var,BBSAMC} it is presented a characterization of cyclic and BCH codes whose apparent distance reaches their minimum distance. Our aim is to extend those results for multivariate codes. 

\begin{theorem}\label{caracterizacion codigos multi distancia real}
 Let $C$ be an abelian code in $\F(r_1,\dots,r_s)$. The following conditions are equivalent:
 \begin{enumerate}
  \item $\Delta_{\B}(C)=d(C)$.
  \item There exist an element $\overline{\alpha}\in U$ and a codeword $c\in C$  such that its image under the discrete Fourier transform, $g=\varphi_{\overline{\alpha},c}$, verifies:
  \begin{enumerate}
   \item $\Delta_{\B}\left(M(g)\right)=\Delta_{\B,\overline{\alpha}}(C)=\min\left\{\Delta_{\B}\left(M(\varphi_{\overline{\alpha},v})\right)\tq v\in C\right\}$. 
   
   \item $\Delta_{\B}\left(M(g)\right)=\left|\overline{Z(g)}\right|$.
  \end{enumerate}
 \end{enumerate}
\end{theorem}

\begin{proof}
 $[1\Rightarrow 2]$ Let $c\in C$ a be codeword with $\omega(c)=d(C)$ and  $\overline{\alpha}\in U$  such that 
$\Delta_{\B,\overline{\alpha}}(C)=\Delta_{\B}(C)$. Set $g=\varphi_{\overline{\alpha},c}$. Then 
 \[d(C)=\Delta_{\B}(C)= \Delta_{\B,\overline{\alpha}}(C)\leq\Delta_{\B}\left(M(g)\right)\leq \left|\overline{Z(g)}\right|=\omega(c)=d(C),\]
 where the first equality is given by hypothesis. Thus, the inequality becomes equality.
 
 $[2\Rightarrow 1]$ Suppose that there is a codeword $c\in C$ satisfying the hypotheses. Then
 \[d(C)\leq \omega(c)=\left|\overline{Z(g)}\right|=\Delta_{\B}\left(M(g)\right)=\Delta_{\B,\overline{\alpha}}(C)\leq\Delta_{\B}(C)\leq d(C)\]
 and, again, equalities hold.
\end{proof}

\begin{remarks}\label{remark sobre el teorema de caracterizacion}
\rm{ The conditions in statement \textit{(2)} of Theorem~\ref{caracterizacion codigos multi distancia real} refers only to a single element in $U$. This is an important reduction that will be very useful later.  On the other hand, we recall that if $M$ is the hypermatrix afforded by $\D_{\overline{\alpha}}(C)$ then $\Delta_{\B,\overline{\alpha}}(C)=\B\mad(M)$. 
}\end{remarks}

For a given abelian code, the problem of finding, if any, a codeword verifying Condition \textit{(2)} of Theorem~\ref{caracterizacion codigos multi distancia real} is in general difficult to solve. In the case that the codeword is an idempotent we will be able to find it \textit{through the computation of the minimum apparent distance}.  So far we only know an actual way to find the desired idempotent codeword of the mentioned theorem and it is given in the following result.

\begin{proposition}\label{condicion suficiente para distancia real}
Let $C$ be a code in $\F(r_1,r_2)$, $\overline{\alpha}\in U$ and $M$ the matrix afforded by its defining set 
$\mathcal{D}_{\overline{\alpha}}(C)$. Let $P\leq M$ be a $q$-orbits matrix and $e\in \mathbb{L} (r_1,r_2)$ be the idempotent such that that $P=M(e)$. If $P$ verifies
\begin{enumerate}
 \item  $\Delta(P)_{\B} = \B\mad(M)$ and
 \item  $\Delta(P)_{\B}=\left|\overline{Z(e)}\right|$.
\end{enumerate}
Then $d(C)= \B\mad(M)=\Delta_{\B}(C)$.
\end{proposition}
\begin{proof}
Since  $P\leq M$, then $\varphi^{-1}_{\overline{\alpha},e}\in C$, hence $\omega(\varphi^{-1}_{\overline{\alpha},e})\geq d(C) \geq \Delta_{\B}(C)$. On the other hand, by the hypothesis \textit{(2)},  
$$\omega(\varphi^{-1}_{\overline{\alpha},e})=\Delta_{\B}(P) = \B\mad(M)\leq \Delta_{\B}(C)\leq d(C).$$
Therefore, $d(C)= \B\mad(M)=\Delta_{\B}(C)$.
\end{proof}

So, for a given code $C$ with afforded matrix $M$, if we want to know whether $d(C)=\Delta_{\B}(C)$ by using 
Proposition~\ref{condicion suficiente para distancia real}, in the case that the requested codeword is an idempotent, we have to analyze all ($q$-orbits) matrices $P\leq M$. If $\left|\overline{\mathcal{D}_{\overline{\alpha}}(C)}\right|=t$ and such idempotent exists, we have to do at most $2^t$ steps. This is a search on the set of idempotents of $C$. The reader may note that the original computation of the apparent distance in \cite{Camion} and \cite{Evans} requires to compute  the apparent distance of \textit{exactly} the same set of $q$-orbits matrices. This might be an important reduction in some cases.

We wonder if it is possible to simplify the procedure to find a $q$-orbits matrix $P$, as in Proposition~\ref{condicion suficiente para distancia real}, by analyzing the sequence of matrices in the algorithm for the computation of the minimum apparent distance of the matrix $M$ afforded by $\mathcal{D}_{\overline{\alpha}}(C)$; i.e. the computation of $\B\mad(M)$. The algorithm gives us an interesting reduction. 

In the algorithm for the computation of the strong apparent distance as in Theorem~\ref{teodam2}, we consider the sequence of matrices 
\begin{equation} \label{sequencia de matrices}
M=M_0 > M_1> \cdots > M_{j_0-1}>  M_{j_0}> \cdots >M_{\ell}> 0
\end{equation}
and let $j_0$ be the first index such that $\Delta_{\B}\left(M_{j_0}\right)= m_{\ell}=\B\mad(M)$. If $m_0=\Delta_{\B}(M)$ equals $m_{\ell}$, then $P= M$ and we do not have any reduction. However, if $m_0>m_{\ell}$,  then $\Delta_{\B}(P)=m_{\ell} <m_{j_0-1}$ which implies $P\leq M_{j_0}$, hence we can start our search from $M_{j_0}$; 
that is, we have to check only at most $2^{t-j_0}$ matrices in order to find the hypothetical matrix of Theorem~\ref{caracterizacion codigos multi distancia real}

We wonder if the existence of a matrix $P\leq M$ satisfying the conditions of Proposition~\ref{condicion suficiente para distancia real} implies the existence of a matrix in the sequence \eqref{sequencia de matrices} also satisfying those conditions. The answer is negative, as the following very simple example shows.

\begin{example} \label{ex36}
 Set $\Delta=\Delta_{\B}$, with $\B=\{\delta_{BCH}\}$. There exists an abelian code $C$, with matrix $M$ afforded by $\mathcal{D}_{\overline{\alpha}}(C)$ with respect to $\overline{\alpha}\in U$ such that:
 \begin{enumerate}
  \item For every $q$-orbits matrix in the sequence $M=M_0 > \cdots > 0$ we have $\Delta(M_j)\neq \left|\overline{Z(e_j)}\right|$, where $e_j\in \Le(r_1,r_2)$ is the idempotent that verifies $M_j=M(e_j)$.
  \item $d(C)=\Delta(C)$
 \end{enumerate}
\end{example}
\begin{proof}
 Set $q=2$ and $r_1=r_2=7$. Let $C$ be the code such that $\mathcal{D}_{\overline{\alpha}}(C)=Q(0,3)\cup Q(1,3)\cup Q(1,5)\cup Q(1,6) \cup Q(3,0)\cup (3,2)\cup Q(3,3)\cup Q(3,4)\cup Q(3,5)\cup Q(3,6)$ with respect to $\overline{\alpha}\in U$. The matrix afforded by $\mathcal{D}_{\overline{\alpha}}(C)$ is 
 $$
M=\left(
\begin{array}{ccccccc}
1 & 1 & 1 & 0 & 1 & 0 & 0 \\
1 & 1 & 1 & 0 & 1 & 0 & 0 \\
1 & 1 & 1 & 0 & 1 & 0 & 0 \\
0 & 1 & 0 & 0 & 0 & 0 & 0 \\
1 & 1 & 1 & 0 & 1 & 0 & 0 \\
0 & 0 & 0 & 0 & 1 & 0 & 0 \\
0 & 0 & 1 & 0 & 0 & 0 & 0 
\end{array}
\right).
$$

Let  $a(X_1)= (1+X_1)(1+X_1^2+X_1^3)$, $b(X_2)=(1+X_2)(1+X_2^2+X_2^3)$. If $e\in C$ is the idempotent generator then $\varphi_{\overline{\alpha},e}(X_1,X_2)=  a(X_1)\,b(X_2) + X_1^3X_2+ X_1^6X_2^2+X_1^3X_2^4$. One may compute   
$\left|\overline{Z(\varphi_{\overline{\alpha},e})}\right|=25$, by using GAP.
 
 On the other hand, computing $\B\mad(M)$, we obtain the chain $M_0>0$ and $\Delta(M)=\B\mad(M)=9$. Now consider 
the $q$-orbits matrix
 $$
P=M\left(ab\right)=\left(
\begin{array}{ccccccc}
1 & 1 & 1 & 0 & 1 & 0 & 0 \\
1 & 1 & 1 & 0 & 1 & 0 & 0 \\
1 & 1 & 1 & 0 & 1 & 0 & 0 \\
0 & 0 & 0 & 0 & 0 & 0 & 0 \\
1 & 1 & 1 & 0 & 1 & 0 & 0 \\
0 & 0 & 0 & 0 & 0 & 0 & 0 \\
0 & 0 & 0 & 0 & 0 & 0 & 0 
\end{array}
\right).
$$

Note that $P$ does not belong to the sequence $M_0>0$; however, as $a\mid X_1^7-1$ and $b\mid X_2^7-1$, then $g_1=ab$ satisfies the hypothesis of Proposition~\ref{g=ab} which means that $C$ satisfies the condition \textit{(2b)} of Theorem~\ref{caracterizacion codigos multi distancia real}. Now, as $\Delta(P)=9$ then condition \textit{(2a)} of the same theorem is satisfied and thus $d(C)=\Delta(C)$.
\end{proof}

As we have seen, although Proposition~\ref{condicion suficiente para distancia real} gives us a sufficient condition, it does not guarantee that we can find the desired codeword if it is not an idempotent, not even by using the algorithm. Now, in order to construct codes $C$ satisfying $d(C)=\Delta_{\B}(C)$, we try to move forward into a different direction: we firstly characterize those polynomials that verify \textit{(2b)}; that is, $\Delta_{\B}\left(M(g)\right)=|\overline{Z(g)}|$. In the univariate case, those polynomials were characterized in \cite{BBSAMC}. Before to extend the results to multivariate polynomials, we need to put some restrictions on the election of ds-bounds that we may use. In the case of polynomials in one variable, one may see that the condition \textit{(2b)} in Theorem~\ref{caracterizacion codigos multi distancia real} forces us to use exclusively the BCH bound as, to the best of our knowledge, the computation of $|\overline{Z(g)}|$ is only known to be obtained in terms of the degree of $\overline{X^hg}$ (viewed as a polynomial); that is, a list of consecutive exponents of the monomials with coefficient zero of the highest degrees.

In the reminder of this section, consider  $\B=\{\delta_{BCH}\}$ and  denote $\Delta=\Delta_{\delta_{BCH}}$, for the sake of simplicity. Let us recall some facts from univariate polynomials that will be used herein. For $0\neq g\in \mathbb{L} (n)$, let $m_g=\gcd (X^{n}-1, X^hg)$, which does not depend on $h\in \N$. As we pointed out in Remark~\ref{d* y DeltaBCH} about notation, $sd^*=\Delta_{\{\delta_{BCH}\}}=\Delta$.  The proof of the following result is essentially the same as~\cite[Proposition 1]{BBSAMC}; so we ommit it.

\begin{proposition} \label{1varcond}
Let $g\in \Le(n)$. Then 

$\Delta(M(g))=\left|\overline{Z(g)}\right|$  if and only if $\overline{X^hg} \mid X^n-1$,  for some $ h\in \mathbb{N}$.
\end{proposition}

Consider $g=g(X_1,X_2)\in \mathbb{L} (r_1,r_2)$ and write $M=M(g)$. In general, as we have seen in Theorem~\ref{boundDFTnvar},  $\Delta(M(g)) \leq \left|\overline{Z(g)}\right|$. We want to describe the polynomials $g$ in $\mathbb{L} (r_1,r_2)$  such that $\Delta(M(g)) = \left|\overline{Z(g)}\right|$ so we assume that the equality holds; moreover, we impose the following condition
\begin{equation} \label{conditionimposed}
\Delta_1(M)=\Delta_2(M)=\Delta(M)=\left|\overline{Z(g)}\right|.
\end{equation}
where, as in Definition~\ref{Bappdistmatrix}, $\Delta_{i}(M)=\omega_{i}(M)\cdot \epsilon_{i}(M)$, for $i\in\{1,2\}$.

We also write, following the notation of paragraph prior to Theorem~\ref{boundDFTnvar},

\begin{equation} \label{gXYYX}
g=[g(X_1)](X_2)=\sum_{k=0}^{r_2-1}g_{2,k}X_2^k , \;  \; g=[g(X_2)](X_1)=\sum_{k=0}^{r_1-1}g_{1,k}X_1^k .
\end{equation} 

For all $u\in \mathbb{L}$, the polynomials of the form $g(u,X_2)$ and $g(X_1,u)$  have the obvious meaning. For  $j=1,2$, 
denote by $\overline{Z}_j= \pi_j\left(\overline{Z(g)}\right)$ the projection of $\overline{Z(g)}$ 
onto the $j$-th coordinate and we also write ${Z}_j= \pi_j\left({Z(g)}\right)$.

For each $j\in\{1,2\}$, we set $M_j=\left\{k\in\{0,\dots,r_i-1\}\tq (j,k)\in I_p(M)\right\}$; that is, $\Delta\left(M(g_{j,k})\right)=\varepsilon_j$, for any $k\in M_j$. Note that $I_p(M)=(1,M_1)\cup (2,M_2)$, (see Definition~\ref{indices involucradas}).

Now,  for each  $k\in M_1$, define 
\[ D_{1,k}  =\left\{(u,v)\in R \tq v\in \overline{Z(g_{1,k})}\;\text{ and  }\; u\in \overline{Z\left(g(X_1,v)\right)} \right\} \]
and analogously, for each $k\in M_2$, define
\[ D_{2,k}=\left\{(u,v)\in R \tq u\in \overline{Z(g_{2,k})}\;\text{ and  }\; v\in \overline{Z\left(g(u,X_2)\right)} \right\}. \]
So, if we set, for $j\in \{1,2\}$, $\overline{D_{j,k}}=\min\left\{\left|\overline{Z(g(u,X_j))}\right|\tq u \in \overline{Z(g_{j,k})} \right\}$ then $\left|\overline{Z(g_{j,k})}\right|\cdot \left|\overline{D_{j,k}}\right|\leq \left|D_{j,k}\right|$. 

We know that, for each $k\in M_1$, it happens  $\varepsilon_1(M)=\Delta\left(M(g_{1,k})\right)\leq  \left|\overline{Z\left(g_{1,k}\right)}\right|$ and for each $v \in \overline{Z\left(g_{1,k}\right)}$, we have $\omega_1(M)\leq\Delta(M(g(X_1,v)))$ and so $\omega_1(M)\leq  \left|\overline{D_{1,k}}\right|$, hence $\varepsilon_1(M)\omega_1(M)\leq \left|D_{1,k}\right|$. However, by the condition~\eqref{conditionimposed}, $\varepsilon_1(M)\omega_1(M)=\Delta_1(M)= \left|\overline{Z(g)}\right|$ and, by definition, $\left|D_{1,k}\right| \leq \left|\overline{Z(g)}\right|$. Therefore, for all $k\in M_1$,
\begin{equation} \label{DeltaZg}
\left|\overline{Z(g)}\right|=\varepsilon_1(M)\omega_1(M)= \left|D_{1,k}\right|\quad\text{and so}\quad D_{1,k}=\overline{Z(g)}.
\end{equation}
Analogously, $D_{2,k}=\overline{Z(g)}$, for $k\in M_2$.

In fact, for $j\in\{1,2\}$ and any $k\in M_j$, we have $\left|\overline{Z(g_{j,k})}\right|\cdot \left|\overline{D_{j,k}}\right|= \left|D_{j,k}\right|$,  hence
\begin{equation}\label{epsilon es Z y omega D barra}
\varepsilon_1(M)=\Delta\left(M(g_{1,k})\right)=\left|\overline{Z\left(g_{1,k}\right)}\right|  \;\text{ and  }\; \omega_1(M) =  \left|\overline{D_{1,k}}\right|.
\end{equation}

Keeping in mind the equalities obtained in the previous paragraphs, we get the following two results.

\begin{lemma} \label{lema1} 
Let $g =g(X_1,X_2)\in \mathbb{L} (r_1,r_2)$ be a polynomial such that $M=M(g)$ satisfies the 
condition~\eqref{conditionimposed}. Then:
\begin{enumerate} 
 \item \label{distappvar2} For each $k\in M_1$, $\Delta\left(M(g_{1,k})\right)=  \left|\overline{Z\left(g_{1,k}\right)}\right|=\left|\overline{Z_2}\right|$ and $\Delta(M(g(X_1,v)))=  \left|\overline{Z\left(g(X_1,v)\right)}\right|$, for any 
$v\in \overline{Z\left(g_{1,k}\right)}$.

\item \label{distapp1var}  For each $k\in M_2$, $\Delta\left(M(g_{2,k})\right)=  \left|\overline{Z\left(g_{2,k}\right)}\right|=\left|\overline{Z_1}\right|$ and $\Delta(M(g(u,X_2)))=  \left|\overline{Z\left(g(u,X_2)\right)}\right|$, for any 
$u\in \overline{Z\left(g_{2,k}\right)}$.
 
\end{enumerate}
\end{lemma}

\begin{proof} We prove \textit{(1.)} as the the proof of \textit{(2.)} is entirely analogous. As we have already seen, if condition~\eqref{conditionimposed} is satisfied then (\ref{DeltaZg}) and \eqref{epsilon es Z y omega D barra} also hold; so that, $\varepsilon_1(M)=\Delta\left(M(g_{1,k})\right)=\left|\overline{Z\left(g_{1,k}\right)}\right|$.

Once we have the first equality, if $\omega_1(M)=\Delta\left(M(g(X_1,v))\right) < \left|\overline{Z\left(g(X_1,v)\right)}\right|$, for some $v \in \overline{Z\left(g_{1,k}\right)}$, then it must happen $\varepsilon_1(M)\omega_1(M)< \left|D_{1,j}\right|$, a contradiction.  Finally, if $v\in \overline{Z_2}$, then there exists $u\in R_{r_1}$ such that $(u,v)\in \overline{Z(g)}=D_{1,k}$, hence $u\in \overline{Z\left(g_{1,k}\right)}$. This proves the other equalities of this lemma. 
\end{proof}

\begin{proposition} \label{gabF}
Let $g=g(X_1,X_2)\in \mathbb{L} (r_1,r_2)$ be a polynomial such that $M=M(g)$ satisfies the condition~\eqref{conditionimposed}. Then there exist $a=a(X_1) \in \mathbb{L} (r_1), \, b=b(X_2)\in \mathbb{L} (r_2)$ and 
$F=F(X_1,X_2)\in \mathbb{L} (r_1,r_2)$ such that 
$ g=abF$ and
\begin{enumerate}
 \item $\overline{X_1^{h_1}a} \mid (X_1^{r_1}-1)$, for some $h_1\in \Z_1$, with $\Delta(M(a))=\varepsilon_2(M)$.
 \item $\overline{X_2^{h_2}b} \mid (X_2^{r_2}-1)$, for some $h_2\in \Z_2$, with  $\Delta(M(b))=\varepsilon_1(M)$.
\end{enumerate}
\end{proposition}

\begin{proof}
By Lemma~\ref{lema1}.1 and by Proposition~\ref{1varcond}, for each $k\in M_2$, if we denote $m_{k}=\gcd\left(g_{2,k},X_1^{r_1}-1\right)$, then 
\begin{eqnarray*}
 \Delta(M(m_k))&=&\left|\overline{Z(m_k)}\right|=r_1-\left|Z(m_k)\right|=r_1-\left|Z\left(g_{2,k}\right)\right|=\\
 &=&\left|\overline{Z\left(g_{2,k}\right)}\right|=\Delta\left(M(g_{2,k})\right).
\end{eqnarray*}

By definition, $\Delta\left(M(g_{2,k})\right) = r_1-\deg (\overline{X_1^{k'} g_{2,k}})$, for some $k'\in \N$. As $\Delta(M(m_k))=\Delta\left(M(g_{2,k})\right)$ and, by \cite[Lemma 2]{BBSAMC},
$\Delta(M(m_k))=r_1-\deg m_k$, then  $\overline{X_1^{k'} g_{2,k}}$ and $m_k$ are associated.

Now we claim that $m_k\mid g_{2,j}$, for all $j\in \{ 0, \ldots, r_2-1 \}$. Indeed, for a fixed 
 $k\in M_2$, by~\eqref{DeltaZg}, we have $D_{2,k}=\overline{Z(g)}$ which implies $\overline{Z\left(g_{2,j}\right)}\subseteq \overline{Z\left(g_{2,k}\right)}$ or rather $Z\left(g_{2,k}\right) 
\subseteq Z\left(g_{2,j}\right)$. Hence, $m_k\mid g_{2,j}$, for all $j\in \{ 0,\ldots,r_2-1\}$. 

Denoting by $g'_{2,j}=\frac{g_{2,j}}{m_k}$, for all $j\in \{ 0,\ldots,r_2-1\}$ and   $a(X_1)=m_k$, we may write
\begin{equation}\label{gXY}
 g(X_1)(X_2)=a(X_1)\sum_{j=0}^{r_2-1}g'_{2,j}X_2^j,
\end{equation}
with $\Delta(M(a(X_1)))=\varepsilon_2(M)$. 

Analogously, we get
\begin{equation}\label{gYX}
 g(X_2)(X_1)=b(X_2)\sum_{i=0}^{r_1-1}g'_{1,i}X_1^i,
\end{equation}
with $\Delta(M(b(X_2)))=\varepsilon_1(M)$ and $b(X_2)g'_{1,i}=g_{1,i}$, for any  $i\in \{ 0,\ldots,r_1-1\}$. 

It is important to note that $1=\gcd \left( X_1^{r_1}-1, g'_{2,0},\ldots, g'_{2,r_2-1} \right)$ and
$1=\gcd \left( X_2^{r_2}-1, g'_{1,0},\ldots,  g'_{1,r_1-1} \right)$.

Now by writing
$$f(X_1,X_2)=\sum_{j=0}^{r_2-1}g'_{2,j}X_2^j\quad\text{ and }\quad h(X_1,X_2)=\sum_{i=0}^{r_1-1}g'_{1,i}X_1^i,$$
we get $g(X_1,X_2)=a(X_1)f(X_1,X_2)=b(X_2)h(X_1,X_2)$.

Recall that ${Z}_1= \pi_1\left({Z(g)}\right)$ and ${Z}_2= \pi_2\left({Z(g)}\right)$.

First note that if $v\in \overline{Z_2}$, then there exists $u\in \overline{Z_1}$ such that
$(u,v)\in \overline{Z(g)}$. This implies $(u,v)\in D_{1,k}$, for $k\in M_1$, 
hence $v\in \overline{Z(g_{2,k})}=\overline{Z(b)}$.
Therefore, $\overline{Z_2} \subseteq \overline{Z(b)}$ and, by Lemma~\ref{lema1}.\ref{distappvar2}, 
$\overline{Z_2} = \overline{Z(b)}$.

Consider $g(X_1,v)= b(v)\,h(X_1,v)= b(v) \sum_{i=0}^{r_1-1}g'_{2,i}(v)X_1^i$, for any $v\in R_{r_2}$. If $b(v)\neq 0$, then $g(X_1,v)\neq 0$, otherwise all $g'_{2,i}$ would have a common zero, which is not possible. Conversely, 
if $g(X_1,v)\neq 0$, then $b(v)\neq 0$. This proves that $v\in \overline{Z(b)}=\overline{Z_2}$ if and only if $g(X_1,v)\neq 0$.

Obviously $g(X_1,v)= 0$ also implies $b(v)= 0$, hence $v\in Z_2$, as $Z(b) \subseteq Z_2$. Now let us write

\[f(X_1,X_2)=\sum_{i=0}^{r_1-1}f_{2,i}X_1^i \quad \text{ and }\quad h(X_1,X_2)= \sum_{j=0}^{r_2-1}h_{1,j}X_2^j.\]

If $v\in Z(b)$, then $g(X_1,v)=0$ and we have $f(X_1,v)=0$. Since $a(X_1)\neq 0$ we must have $f_{2,i}(v)=0$, for all $i\in \{ 0,\ldots,r_1-1\}$. Hence $Z(b) \subset Z(f_i)$ and $b(X_2)\mid f_{2,i}$, for all $i\in \{ 0,\ldots,r_1-1\}$. Now if $f_{2,i}(v)=0$, for all $i\in \{ 0,\ldots,r_1-1\}$, then $f(X_1,v)=0$ and  $g(X_1,v)=0$, which implies $v\in Z(b)$, as we have seen before. Hence $b(X_2)=\gcd (X_2^{r_2}-1, f_{2,i})$, for all $i\in \{ 0,\ldots,r_1-1\}$. Therefore, $f(X_1,X_2)= b(X_2) f'(X_1,X_2)$ and 

\begin{equation}  \label{gabf1}
	g(X_1,X_2)=a(X_1)b(X_2)f'(X_1,X_2).
\end{equation}

Analogously, one may prove that $a(X_1)\mid h_{1,j}$, for all $j\in \{ 0,\ldots,r_2-1\}$, and get $h(X_1,X_2)= a(X_1) h'(X_1,X_2)$, hence 	
	\begin{equation}  \label{gabf2}
 g(X_1,X_2)=a(X_1)b(X_2)h'(X_1,X_2).
	\end{equation}

 Finally, note that the decompositions $g=abf'$ and $g=abh'$ from~\eqref{gabf1} and~\eqref{gabf2}, has been done in $\Le[X_1,X_2]$, which is a domain, and so, we have $f'(X_1,X_2)=h'(X_1,X_2)$. By writing $F(X_1,X_2)=f'(X_1,X_2)=h'(X_1,X_2)$, we get 
\begin{equation*} 
  g(X_1,X_2)=a(X_1)\, b(X_2) \, F(X_1,X_2).
 \end{equation*} 
\end{proof} 
 
It is clear that the condition~\eqref{conditionimposed} plays an important role in all the previous proofs. Recall that a polynomial $g\in \Le(r_1,r_2)$, with coefficient matrix $M=M(g)$, satisfies such condition if 
$$\Delta_1(M)=\Delta_2(M)=\Delta(M)=\left|\overline{Z(g)}\right|.$$

For those polynomials, we have obtained a factorization $g=abF$, which describes them, where $\Delta(M(a))=\varepsilon_2(M)$ and $\Delta(M(b))=\varepsilon_1(M)$. 

At this point, two questions arise for an abelian code satisfying Theorem~\ref{caracterizacion codigos multi distancia real} with a codeword image $g=\varphi_{\overline{\alpha},c}$ as in such theorem.
\begin{enumerate}
 \item Is it always true that $g$ satisfies also condition~\eqref{conditionimposed}?
 \item Suppose a polynomial $g\in \Le(r_1,r_2)$ already satisfies condition~\eqref{conditionimposed} and so we have a decomposition $g=abF$. What can we say about $F$? More specifically, is it true that $M(F)$ is a $q$-orbits matrix? And is it true that $\Delta(M(a))=\Delta(M(b))=\Delta(M(g))$?
\end{enumerate}

We shall answer all these questions in the following examples.

\begin{example} \label{ex32} 
There exists an abelian code $C$ generated by an idempotent $e \in \F(r_1,r_2)$, with image $g=\varphi_{\overline{\alpha},f}$, satisfying the following properties.
\begin{enumerate}
 \item $\Delta_1(M(g))<\Delta_2(M(g))$ (so the condition~\eqref{conditionimposed} is not fully satisfied)
 \item $\Delta(M(g))=\left|\overline{Z(g)}\right|$.
 \item $d(C)=\Delta(C)$.
\end{enumerate}
 \end{example}
\begin{proof}
 Set $q=2$, $r_1=5$, $r_2=9$ and $C$ be the code with $\overline{D(C)}=Q(1,3)$, a minimal code with generator idempotent $e(X_1,X_2)=X_1^4X_2^7+X_1^3X_2^8+X_1^4X_2^6+X_1^2X_2^8+X_1^3X_2^6+X_1^4X_2^4+X_1^3X_2^5+X_1^2X_2^6+X_1X_2^7+X_1^4X_2^3+ X_1^2X_2^5+X_1X_2^6+ X_1^3X_2^3+X_1^4X_2+X_1^3X_2^2+X_1^2X_2^3+ X_1X_2^4+ X_1^4+ X_1^2X_2^2+X_1X_2^3+X_1^3+ X_1^2+X_1X_2+X_1$.  Using the program GAP, we computed 
$d(C)=24$. One may check that
$\varphi_{\alpha,e}=g(X_1,X_2)=X_1X_2^3+X_1^4X_2^3+X_1^2X_2^6+X_1^3X_2^6$. Some direct computations yield  $\Delta_1(M(g))=18$ and $\Delta(M(g))=24$, so assertion \textit{(1)} of this example is satisfied. As $\omega(e)=24$ we also get assertion \textit{(2)}.  Since $C$ is minimal, we have, $24=\Delta(M(g))=\B\mad(M(g))$. On the other hand $\Delta(M(g))\leq \Delta(C)\leq d(C)=24$. Thus $d(C)=\Delta(C)$ and we get assertion \textit{(3)}.
\end{proof}

The previous example gives a negative answer to question 1. The following example answers question 2.

\begin{example} \label{ex33}
Under the same notation of Proposition~\ref{gabF}, there exists an abelian code $C$ generated by an 
idempotent $e \in \F(r_1,r_2)$, with image $g=\varphi_{\overline{\alpha},e}$, 
such that the following properties hold.
\begin{enumerate}
 \item $M(g)$ satisfies the condition~\eqref{conditionimposed} and then $g=abF$ as in the mentioned proposition (see paragraph prior to Example~\ref{ex32}).
 \item $d(C)=\Delta(C)$. 
 \item $F(X_1,X_2)$ has at least two nonzero monomials and $M(F)$ is not a $q$-orbits matrix.
 \item $\Delta(a)=\Delta(b)$,  but $\Delta(a)\Delta(b)\neq\Delta(M(g))$
 \item $\overline{Z(g)} \neq \overline{Z_1}\times \overline{Z_2}$ 
 \end{enumerate}
\end{example}
\begin{proof}
 Let $q=2$, $r_1=r_2=5$ and $C$ be the code with  $\overline{D(C)}=Q(1,1)\cup Q(1,3)$. 
In this case, $g(X_1,X_2)=X_1^4X_2^4+X_1^3X_2^4+X_1^4X_2^2+X_1^3X_2^3+X_1^2X_2^2+X_1X_2^3+X_1^2X_2+X_1X_2$ and $\varphi^{-1}_{\alpha,g}=X_1^3X_2^4+X_1^4X_2^2+X_1^4X_2+X_1^3X_2^2+X_1^2X_2^3+
X_1X_2^4+X_1X_2^3+X_1^2X_2$, so that $|\overline{Z(g)}|=8$.

By using GAP, we computed $d(C)=8$ and one may check that $\Delta_1(M(g))=\Delta_2(M(g))=\Delta(M(g))=8$; so that \textit{(1.)} and \textit{(2.)} hold.

Let us factorize $g$. In the case 
  $$g(X_1)(X_2)=(X_1+X^2_1)X_2+(X^2_1+X_1^4)X_2^2+(X_1+X_1^3)X_2^3+(X_1^3+X_1^4)X_2^4,$$
  $M_2=\{1,4\}$ and $a=a(X_1)=1+X_1$. Note that $(1+X_1)X_1$ is a common factor in all summands of $g(X_1)(X_2)$.
    On the other hand, 
  $$g(X_2)(X_1)=(X_2+X_2^3)X_1+(X_2+X_2^2)X_1^2+(X_2^3+X_2^4)X_1^3+(X_2^2+X_2^4)X_1^4,$$
  $M_1=\{2,3\}$ and $b=b(X_2)=1+X_2$. Here  $(1+X_2)X_2$ is a common factor in all summands of $g(X_2)(X_1)$.
Moreover,
	\begin{eqnarray*}
   f(X_2)(X_1)&=&(X_2+X_2^3)X_1+(X_2^2+X_2^3)X_1^2+(X_2^2+X_2^4)X_1^3\\
   h(X_1)(X_2)&=& (X_1+X_1^2)X_2+(X_1+X_1^4)X_2^2+(X_1^3+X_1^4)X_2^3\quad\text{and so }\\
   F(X_1,X_2)&=& X_1X_2+X_1X_2^2+X_1^2X_2^2+X_1^3X_2^2+X_1^3X_2^3.
  \end{eqnarray*}
This gives us \textit{(3.)}

Now, one may easily check that $\Delta(M(a))=\Delta(M(b))=4$, hence $\Delta(M(a))\cdot\Delta(M(b))\neq \Delta(M(g))$. This gives us \textit{(4)}. 

Finally, by using GAP we compute  $\overline{Z(g)}=Q(1,3)\cup Q(1,4)$ and clearly $\overline{Z(g)} \neq \overline{Z_1}\times \overline{Z_2}$. 
\end{proof}

To finish our argumentation from Theorem~\ref{caracterizacion codigos multi distancia real} we prove that a polynomial that satisfies condition~\eqref{conditionimposed}, and so factorizes  $g=abF$ with $F(X_1,X_2)$ a monomial in $\Le(r_1,r_2)$, verifies that its image under the discrete Fourier transform satisfy condition \textit{(2)} of the mentioned theorem.

\begin{proposition}\label{g=ab}
 Suppose $g\in \mathbb{L}(r_1,r_2)$ is such that $g(X_1,X_2)=a(X_1)\, b(X_2)$, where $a$ and $b$ satisfy Proposition~\ref{1varcond}. Set $M=M(g)$. Then 
 \begin{enumerate}
  \item $\overline{Z(g)}=\overline{Z_1}\times \overline{Z_2}$
  \item $\Delta(M)=\Delta(M(a))\cdot \Delta(M(b))=\left|\overline{Z(g)}\right|$.
  \item $\Delta_1(M)=\Delta_2(M)=\Delta(M)=\left|\overline{Z(g)}\right|$ (the condition~\eqref{conditionimposed}).
  \item $\Delta(M(a))=\varepsilon_2(M)=\omega_1(M)$ and 
 \item $\Delta(M(b))=\varepsilon_1(M)=\omega_2(M)$.
 \end{enumerate}
\end{proposition}

\begin{proof}
Assertion \textit{(1)} and the equality $\Delta(M(a))\cdot \Delta(M(b))=\left|\overline{Z(g)}\right|$ come directly from the decomposition of $g$ together with the hypothesis that $a$ and $b$ satisfy Proposition~\ref{1varcond}.

Now set $M=M(g)$. Since $g=ab$ then $H_m(2,j)=M(a)$, for all $j\in\supp_2(M)$ and $\supp_2(M)=\supp(M(b))$; so that $\varepsilon_2(M)=\Delta(M(a))$ and $\omega_2(M)=\Delta(M(b))$. Analogously, $\varepsilon_1(M)=\Delta(M(a))$ and $\omega_1(M)=\Delta(M(b))$. From this, we get all assertions.
\end{proof}

To sum up, from Proposition~\ref{gabF} we have obtained what kind of polynomials we have to use to reach condition \eqref{conditionimposed}. This will be the main idea in order to construct abelian codes $C$, with $d(C)=\Delta(C)$. We address this problem in the following section.

\subsection{Application 1: construction of abelian codes for which its multivariate BCH bound, apparent distance and minimum distance coincide.}

 In this section, we continue considering  $\B=\{\delta_{BCH}\}$ and  denoting $\Delta=\Delta_{\delta_{BCH}}$, for the sake of simplicity by the same reasons given in the paragraphs prior Proposition~\ref{1varcond}. Bearing in mind Proposition~\ref{g=ab} and Proposition~\ref{condicion suficiente para distancia real}, we introduce the following definition.
 
\begin{definition}
 A matrix $P$ of order $r_1 \times r_2$, with entries in $\Le$ is called a \textbf{composed polynomial matrix} (\textbf{CP-matrix}, for short) if there exist polynomials $a=a(X_1)\in \Le(r_1)$ and $b=b(X_2)\in\Le(r_2)$ such that $P=M(ab)$, where $ab\in \Le(r_1,r_2)$.
\end{definition}

Note that, for a CP-matrix $P$, its support is a direct product $\supp(P)=\supp(a(X_1))\times \supp(b(X_2))$. The polynomials $a$ and $b$ are called the \textbf{polynomial factors of $P$}. The reader may see that to check if a matrix is a CP-matrix is a trivial task, because it must happen $\pi_1(\supp(P))=\supp(a)=\supp(M(a))$ and $\pi_2(\supp(P))=\supp(b)=\supp(M(b))$. The following result is an immediate consequence of Proposition~\ref{g=ab}.

\begin{corollary}\label{propiedades matrices rectangulares}
Let $P=M(g)$ be a CP-matrix of order $r_1\times r_2$ with polynomial factors $a$ and $b$; that is, $g=ab$. 
If $\overline{X_1^{h_1}a}\mid X_1^{r_1}-1$ and $\overline{X_2^{h_2}b}\mid X_2^{r_2}-1$, for some $h_1, h_2\in \N$, then
 \begin{enumerate}
  \item $\overline{Z(ab)}=\overline{Z(a)}\times \overline{Z(b)}.$
  \item $\Delta(P)=\Delta(M(a))\cdot \Delta(M(b))=\left|\overline{Z(g)}\right|$.
  \item $\Delta_1(P)=\Delta_2(P)=\Delta(P)=\left|\overline{Z(g)}\right|$ (the condition~\eqref{conditionimposed}).
  \item $\Delta(M(a))=\varepsilon_2(P)=\omega_1(P)$ and 
 \item $\Delta(M(b))=\varepsilon_1(P)=\omega_2(P)$.
 \end{enumerate}
\end{corollary}

\begin{example}\label{n=7X3 q=2}
\rm{
 Set $q=2$, $r_1=3$ and $r_2=7$; so that $n=21$. Let $P$ be the CP-matrix with polynomial factors $a=X_1+X_1^2$ and $b=X_2+X_2^2+X_2^4$, and $g=ab$. In this case, $\overline{X_1^2a}
\mid X_1^3-1$ and $\overline{X_2^6 b}\mid X_2^7-1$. Now $\Delta(M(a))=\Delta(M(\overline{X_1^2a}))=2$, 
$\Delta(M(b))=\Delta(M(\overline{X_2^6 b}))=4$; hence $\Delta(P)=8=|\overline{Z(g)}|$.
}\end{example}

The next example shows that the hypothesis on the polynomials $a$ and $b$ of Corollary~\ref{propiedades matrices rectangulares} are not superfluous.

\begin{example}\label{n=5X7 q=2}
\rm{
  Set $q=2$, $r_1=5$ and $r_2=7$; so that $n=35$. Let $P$ be the CP-matrix with factors $a=X_1+X_1^2+X_1^3+X_1^4$ and $b=X_2+X_2^2+X_2^4$. In this case, $\overline{X_1^{h_1}a}\nmid X_1^3-1$, for all $h_1\in\Z_5$. On the other hand $\overline{X_2^6 b}\mid X_2^7-1$. Now $\Delta(M(a))=2$, $\Delta(M(b))=4$. Although $\Delta_1(P)=\Delta_2(P)=\Delta(P)=8$, one may check that $\left|\overline{Z(ab)}\right|=16$.
}\end{example}

Now we give a method for constructing the desired abelian codes. First, a technical lemma.

\begin{lemma}\label{en rectangulares B-mad=da}
 Let $D\subset \Z_{r_1}\times\Z_{r_2}$ be union of $q$-orbits and $M=M(D)$, the matrix afforded by $D$. If $\supp(M)=\pi_1\left(\supp(M)\right)\times \pi_2\left(\supp(M)\right)$ then $\B\mad=\Delta_{\B}(M)$, where $\B$ is any set of ds-bounds.
 
 In the case $\B=\{\delta_{BCH}\}$, the above equality coincides with the multivariate BCH bound in \cite[Theorem 30]{BBCS2}.
\end{lemma}
\begin{proof} 
 Clearly, in this case all rows (columns) have the same support and so if one row or column is involved then all of them are too. The last assertion comes directly from the computation of the multivariate BCH bound.
\end{proof}

\begin{remark}\label{CP=rectang + afforded}\rm{
 We have already mentioned that for any CP-matrix, $M$, one has $\supp(M)=\pi_1\left(\supp(M)\right)\times \pi_2\left(\supp(M)\right)$. The converse is true for those matrices satisfying hypothesis of lemma above; that is, if  $D\subset \Z_{r_1}\times\Z_{r_2}$ is union of $q$-orbits and $M=M(D)$ is the matrix afforded by $D$ with $\supp(M)=\pi_1\left(\supp(M)\right)\times \pi_2\left(\supp(M)\right)$ then $M$ is a CP-matrix.
}\end{remark}

\begin{theorem}\label{construccion con TFD inversa}
Let $\mathbb{K}$ be an intermediate field $\F\subseteq \mathbb{K} \subseteq \Le$, $a=a(X_1)\in \mathbb{K}(r_1)$ and $b=b(X_2)\in \mathbb{K}(r_2)$ be such that $a\mid X_1^{r_1}-1$ and $b\mid X_2^{r_2}-1$. If there exist $(\alpha_1,\alpha_2)\in U$, $h_1\in \Z_{r_1}$ and $h_2\in\Z_{r_2}$  for which $\varphi^{-1}_{\alpha_1,\overline{X_1^{h_1}a}}\in \F(r_1)$ and $\varphi^{-1}_{\alpha_2,\overline{X_2^{h_2}b}}\in \F(r_2)$, then the abelian code 
$C=\left(\varphi^{-1}_{\alpha_1,\overline{X_1^{h_1}a}}\cdot \varphi^{-1}_{\alpha_2,\overline{X_2^{h_2}b}}\right)$ 
in $\F(r_1,r_2)$ verifies $\Delta\left(M(ab)\right)=\Delta(C)=d(C)$.

Moreover, in this case, for any $\beta_1\in U_{r_1}$ and $\beta_2\in U_{r_2}$ the abelian code $C_{(\beta_1,\beta_2)}=\left(\varphi^{-1}_{\beta_1,\overline{X_1^{h_1}a}}\cdot \varphi^{-1}_{\beta_2,\overline{X_2^{h_2}b}}\right)$ is an ideal of $\F(r_1,r_2)$ and verifies $\Delta\left(M(ab)\right)=\Delta(C_{(\beta_1,\beta_2)})=d(C_{(\beta_1,\beta_2)})=d(C)$.
\end{theorem}

\begin{proof}
Set $g(X_1,X_2)=\overline{X_1^{h_1}a(X_1)\cdot X_2^{h_2}b(X_2)}$ and $\overline{\alpha}=(\alpha_1,\alpha_2)$. By definition of the discrete Fourier transform, it is easy to see that the particular factorization of $g$ implies that $\varphi^{-1}_{(\alpha_1,\alpha_2),g}=\varphi^{-1}_{\alpha_1,\overline{X_1^{h_1}a}}\cdot \varphi^{-1}_{\alpha_2,\overline{X_2^{h_2}b}}$. On the other hand, it is clear that $M(g)$ is a CP-matrix satisfying the hypothesis of Corollary~\ref{propiedades matrices rectangulares}. This, in turn, implies that statement \textit{2(b)} of Theorem~\ref{caracterizacion codigos multi distancia real} is satisfied.

Let $M$ be the matrix afforded by $\D_{\overline{\alpha}}(C)$. Since $C=\left(\varphi^{-1}_{\overline{\alpha},g}\right)$ then $\supp(M)=\supp(M(g))$, hence $M$ is also a CP-matrix and $\Delta(M)=\Delta(M(g))$. By Lemma~\ref{en rectangulares B-mad=da}, $\B\mad(M)=\Delta(M)$ and so statement \textit{2(a)} of Theorem~\ref{caracterizacion codigos multi distancia real} is also satisfied. Thus $\Delta(C)=d(C)$. The final assertion is a direct consequence of \cite[Remark 2]{BBSAMC} together with the fact that, under these hypothesis, all afforded matrices are CP-matrices.
\end{proof}

Now, we may apply all known criteria for univariate polynomials to have inverse of the discrete Fourier transform in an specific quotient ring.  The following corollary concretizes the proposed construction. 
It comes from \cite[Remark 2]{BBSAMC} and the theorem above.

\begin{corollary}\label{corolario construccion con a() y b()}
 Let $\mathbb{K}$ be an intermediate field $\F\subseteq \mathbb{K} \subseteq \Le$, $a=a(X_1)\in \mathbb{K}(r_1)$ and $b=b(X_2)\in \mathbb{K}(r_2)$ be such that $a\mid X_1^{r_1}-1$ and $b\mid X_2^{r_2}-1$. If there exist $(\alpha_1,\alpha_2)\in U$, $h_1\in \Z_{r_1}$ and $h_2\in\Z_{r_2}$  for which 
$\left[\left(\overline{X_1^{h_1}a}\right)(\alpha_1^i)\right]^q =\left(\overline{X_1^{h_1}a}\right)(\alpha_1^i)$, 
for all $i\in \{0,\dots,r_1-1\}$ and $\left[\left(\overline{X_2^{h_2}b}\right)(\alpha_2^j)\right]^q =
\left(\overline{X_2^{h_2}b}\right)(\alpha_2^j)$, for all $j\in \{0,\dots,r_2-1\}$, then the 
family of abelian codes 
$$\left\{C_{(\beta_1,\beta_2)}=\left(\varphi^{-1}_{\beta_1,\overline{X_1^{h_1}a}}\cdot 
\varphi^{-1}_{\beta_2,\overline{X_2^{h_2}b}}\right)\tq \beta_1\in U_{r_1}\text{ and }\beta_2\in U_{r_2}\right\}$$
in $\F(r_1,r_2)$ verifies $\Delta\left(M(ab)\right)=\Delta(C_{(\beta_1,\beta_2)})=d(C_{(\beta_1,\beta_2)})$.
\end{corollary}

The following example shows how to use Corollary~\ref{corolario construccion con a() y b()}.

\begin{example}\label{ejemplo 3 por 45}\rm{
 Set $q=2$, $r_1=3$ and $r_2=45$ (so $n=135$). Fix $\alpha_1\in U_3$ and $\alpha_2 \in U_{45}$. Consider the polynomials $a=X+1$ and $b=Y^{40}+Y^{39}+Y^{38}+ Y^{36}+Y^{35}+Y^{32}+Y^{30}+Y^{25}+Y^{24}+Y^{23}+Y^{21}+ Y^{20}+Y^{17}+ Y^{15}+Y^{10}+Y^9+ Y^8+Y^6+Y^5+Y^2+1$. Then $a\mid X^3-1$. Note that $\supp\left(X\cdot a(X)\right)=\{1,2\}=C_2(1)$ modulo $3$. By \cite[Lemma 1]{BBSAMC} we have 
that $h_1=1$ works. Now, the polynomial $b$ appears in \cite[Example 5]{BBSAMC} where it was mentioned that $b\mid x^{45}-1$ in $\mathbb{F}_2[x]$ (so that $\mathbb{K}=\mathbb{F}_2$). In that example, it is shown that, for $\alpha_2 \in U_{45}$ (for instance, the one with minimal polynomial $Y^{12}+Y^3+1$), since $b(1)=1$ and $b(\alpha_2^3)=\alpha_2^{30}$, then $(Y^5b)(1)=1$, $(Y^5b)\left(\alpha_2^3\right)=(\alpha_2^3)^5\alpha_2^{30}=\alpha_2^{45}=1$. So $h_2=5$ will work because $(Y^5b)\left(\alpha_2^6\right)=(Y^5b)\left(\alpha_2^{12}\right)=(Y^5b)\left(\alpha_2^{24}\right)=1$; note that $C_2(3)=\{3,6,12,24\}$ modulo $45$.    Now set $C=(\varphi^{-1}_{\alpha_1,Xa}\cdot \varphi^{-1}_{\alpha_2,Y^5b})\subset \mathbb{F}_{2}(r_1,r_2)= \mathbb{F}_{2}(3,45)$. Then $D_{(\alpha_1,\alpha_2)}(C)=C_2(1)\times \left(C_2(1)\cup C_2(3)\cup  C_2(9)\cup C_2(21)\right)$.

 One  may check that $10=\Delta(M(ab))$; so that $d(C)=10$ and $\dim_{\mathbb{F}_2}(C)=87$.
}\end{example}

The next example shows that from a code satisfying the conditions of Theorem~\ref{caracterizacion codigos multi distancia real}, we can obtain a code with better parameters by making slight modifications on the defining set in such a way that the new code verifies the same conditions, but it has higher dimension, for example.

\begin{example}\label{ejemplo mejorar la dimension 3 por 45}\rm{
 Set $q=2$, $r_1=3$ and $r_2=45$. Fix $\alpha_1\in U_3$ and $\alpha_2 \in U_{45}$. Consider the code $C$ in 
Example~\ref{ejemplo 3 por 45}; that is $D_{(\alpha_1,\alpha_2)}(C)=C_2(1)\times \left(C_2(1)\cup C_2(3)\cup  C_2(9)\cup C_2(21)\right)$ and set $g=Xa\cdot Y^5b$. As one may check there are three subsets determining $\Delta(M(g))$; to witt
 \begin{eqnarray*}
  S_1&=&\{(1,1),\,(1,2),\,(1,3),\,(1,4),\,(2,1),\,(2,2),\,(2,3),\,(2,4)\}, \\ 
  S_2&=&\{(1,16),\,(1,17),\,(1,18),\,(1,19),\,(2,16),\,(2,17),\,(2,18),\,(2,19)\}\;\text{and}\\
  S_3&=&\{(1,31),\,(1,32),\,(1,33),\,(1,34),\,(2,31),\,(2,32),\,(2,33),\,(2,34)\}.
 \end{eqnarray*}
If one computes $\Delta(M(g))$ by considering $S_1$ then clearly $C_2(1)\times C_2(21)$ will have no influence in the computation. Hence one may construct the new code $C'$ for which $D_{(\alpha_1,\alpha_2)}(C')=C_2(1)\times \left(C_2(1)\cup C_2(3)\cup  C_2(9)\right)$ such that  $\Delta(C')=\Delta(C)=\Delta(M(g))$.

Note that the matrix afforded by $D=D_{(\alpha_1,\alpha_2)}(C')$ is also a CP-matrix and so $\B\mad(M(D))=\Delta(M(D))=10$. Since $C$ is a subcode of $C'$ then $c=\varphi^{-1}_{\alpha_1,g}\in C'$ and, clearly, $g$ satisfies conditions \textit{(2a)} and \textit{(2b)} of Theorem~\ref{caracterizacion codigos multi distancia real} for $C'$; hence $\Delta(C')=\Delta(M(g))=d(C')=10=d(C)=\Delta(C)$. Since $\dim_{\mathbb{F}_2}(C')=95$ and $\dim_{\mathbb{F}_2}(C)=87$, $C'$ is a code with better parameters than those of $C$.
}\end{example}

Next application comes from \cite[Corollary 6]{BBSAMC}.

\begin{corollary} \label{construccion distancia verdadera con  irreducibles}
 Let $\mathbb{K}$ be an intermediate field $\F\subseteq \mathbb{K} \subseteq \Le$ and $a=a(X_1)\in \mathbb{K}(r_1)$ 
be such that $a\mid X_1^{r_1}-1$ with $\varphi^{-1}_{\alpha_1,\overline{X_1^{h_1}a}}\in \F(r_1)$, for some $\alpha_1\in U_{r_1}$ and $h_1\in \Z_{r_1}$. 
 
 Let  $g$ be an irreducible factor of $X_2^{r_2}-1$ in $\mathbb{K}[X_2]$ with defining set $D_{\alpha_2}(g)$, 
for some $\alpha_2\in U_{r_2}$. Set $b=(X_2^n-1)/g$.  If there are positive integers $j,t$ such that $b(\alpha_2^j)=\alpha_2^t$ and $\gcd \left(j,\frac{r_2}{\gcd(q-1,r_2)}\right)\mid t$, then there exists $h_2\in \Z_{r_2}$ such that the abelian code  $C=\left(\varphi^{-1}_{\alpha_1,\overline{X_1^{h_1}a}}\cdot \varphi^{-1}_{\alpha_2,\overline{X_2^{h_2}b}}\right)$ in $\F(r_1,r_2)$ verifies $\Delta\left(M(ab)\right)=\Delta(C)=d(C)$.
\end{corollary}

Our last application of this section is the following result that comes from \cite[Corollary 7]{BBSAMC}.

\begin{corollary}
  Let $a=a(X_1)\in \Le(r_1)$ be such that $a\mid X_1^{r_1}-1$ with $\varphi^{-1}_{\alpha_1,\overline{X_1^{h_1}a}}\in \mathbb{F}_2(r_1)$, for some $\alpha_1\in U_{r_1}$ and $h_1\in \Z_{r_1}$, and suppose $r_2=2^m-1$, for some $m\in \N$. Then there exist at least $\frac{\phi\left(r_2\right)}{m}$ binary codes $C$ of length $n=r_1r_2$ such that  $\Delta\left(M(ab
  )\right)=\Delta(C)=d(C)$.
\end{corollary}

In the following examples we take the advantage of the information for cyclic codes  from the tables appearing in \cite{BBSAMC} 

\begin{example}\label{extender codigos a partir de minimales}\rm{
 We first show in Table~\ref{tabela1} some abelian codes of lenght $7\times 15=105$ constructed from a list of divisors $a_i$ of $X_1^{7}-1$ in $\mathbb{F}_2[X_1]$ and divisors $b_j$ of $X_2^{15}-1$ in $\mathbb{F}_2[X_2]$  as in Corollary~\ref{construccion distancia verdadera con  irreducibles}. The divisors are $a_1=1+X_1$, $a_2= 1+X_1+X_1^3$, $a_3=1+X_1^2+X_1^3$, $b_1=\frac{X_2^{15}-1}{1+X_2+X_2^2}$, $b_2=\frac{X_2^{15}-1}{1+X_2+X_2^4}$ and $b_3=\frac{X_2^{15}-1}{1+X_2^3+X_2^4}$.

\begin{table}[h]
\[\begin{array}{|l|l|l|l|l|l|} \hline 
a&h_1& b&h_2&  \text{Dimension} & \Delta=d \\  \hline 
 a_2 & 1 & b_1 & 1   & 30  & 8 \\  \hline 
a_2 & 1 & b_2 & 1   & 24 &  16\\  \hline 
 a_2 & 1 & b_3 & 3   & 24 & 16 \\  \hline 
  a_3 & 3 & b_1 & 1   & 30 & 8 \\  \hline 
a_3 & 3 & b_2 & 1   & 24 &  16  \\  \hline 
 a_3 & 3 & b_3 & 3   & 24 &  16  \\  \hline 
  a_1a_3 & 0 & b_1 & 1   & 40 & 6 \\  \hline 
a_1a_3 & 0 & b_2 & 1   & 32 & 12 \\  \hline 
 a_1a_3 & 0 & b_3 & 3   & 32 & 12 \\  \hline 
   a_2a_3 & 0 & b_1 & 1   & 70 & 2 \\  \hline 
a_2a_3 & 0 & b_2 & 1   & 56 & 4 \\  \hline 
 a_2a_3 & 0 & b_3 & 3   & 56 & 4 \\  \hline 
\end{array}\]

\caption{Abelian Codes in $\mathbb F_2(7,15)$. \label{tabela1}}
\end{table}

In Table~\ref{tabela2} we also have  abelian codes of lenght $105$, but in this case we consider $r_1=5$ and $r_2=21$. 
Here we choose only one divisor of $X_1^{5}-1$ in $\mathbb{F}_2[X_1]$; the $5$-th cyclotomic polynomial, $\Phi_5$. The other divisors $b'_j$ of $X_2^{21}-1$ in $\mathbb{F}_2[X_2]$ from which we may construct abelian codes as in Corollary~\ref{construccion distancia verdadera con  irreducibles} are $b'_1=\frac{X_2^{21}-1}{1+X_2+X_2^2}$, $b'_2=\frac{X_2^{21}-1}{1+X_2+X_2^3}$,  $b'_3=\frac{X_2^{21}-1}{1+X_2^2+X_2^3}$, $b'_4=\frac{X_2^{21}-1}{1+X_2+X_2^2+X_2^4+X_2^6}$ and $b'_5=\frac{X_2^{21}-1}{1+X_2^2+X_2^4+X_2^5+X_2^6}$.
 
\begin{table}[h]
\[\begin{array}{|l|l|l|l|l|l|} \hline 
a&h_1& b&h_2&  \text{Dimension} & \Delta=d \\  \hline 
 \Phi_5 & 0 & b'_1 & 1   &70  & 2 \\  \hline 
\Phi_5 & 0 & b'_2 & 1   & 60 &  3\\  \hline 
 \Phi_5 & 0 & b'_3 & 3   & 60 & 3 \\  \hline 
  \Phi_5 & 0 & b'_4 & 1   & 40 & 6 \\  \hline 
\Phi_5 & 0 & b'_5 & 1   & 40 &  6  \\  \hline 
\end{array}\]

\caption{Abelian Codes in $\mathbb F_2(5,21)$. \label{tabela2}}
\end{table}
}\end{example}

\subsection{Application 2: True distance in BCH multivariate codes}

In \cite[Definition 33]{BBCS2}, the notion of BCH multivariate code appears. Let us recall this definition focused on the bivariate case.

 \begin{definition}
Let $\bar\gamma \subseteq \{1,2\}$ and $\bar\delta=\{(\delta_k)_{k\in \bar\gamma} \tq 2\leq  \delta_k\leq r_k\}$. An abelian code $C$ in $\F(r_1,r_2)$ is a \textbf{bivariate BCH code of designed distance $\bar\delta$} if there exists a list of positive integers $\bar b=(b_k)_{k\in\bar\gamma}$ such that
\[\D_{\overline{\alpha}}(C)=\bigcup_{k\in\bar\gamma} \bigcup_{l=0}^{\delta_k-2}\bigcup_{\bf{i}\in I(k,\overline{b_k+l})}Q(\bf{i})\]
for some $\overline{\alpha} \in U$, where $\{\overline{b_k},\dots,\overline{b_k+\delta_k-2}\}$ is a list of consecutive integers modulo $r_k$ and $I(k,u)=\{\bf i\in I\tq \bf i(k)=u\}$.

The BCH multivariate codes are denoted $B_q(\overline{\alpha},\bar\gamma,\bar\delta,\bar b)$.
\end{definition}

 Let $C$ be an abelian code in $\F(r_1,r_2)$ with $M=M\left(\D_{\overline{\alpha}}(C)\right)$ the matrix afforded by its defining set with respect to some $\bar{\alpha}=(\alpha_1,\alpha_2)\in U$. If $M$ satisfies $\supp\left(M\right)=\pi_1\left(\supp(M)\right)\times \pi_2\left(\supp(M)\right)$ then $\overline{\D_{\overline{\alpha}}(C)}=\pi_1\left(\supp(M)\right)\times \pi_2\left(\supp(M)\right)$. We set $S_1=\pi_1\left(\supp(M)\right)$ and $S_2= \pi_2\left(\supp(M)\right)$.  Then, one may consider the cyclic codes $C_1$ and $C_2$ with defining sets $D_1=\Z_{r_1}\setminus S_1$ and $D_2=\Z_{r_2}\setminus S_2$ with respect to $\alpha_1$ and $\alpha_2$, respectively (note that it may happen $\D_{\overline{\alpha}}(C)\neq D_1\times D_2$).
 
 Now suppose that the code $C$ is an abelian code as described in Theorem~\ref{construccion con TFD inversa} keeping the notation for the polynomials $a$ and $b$ and having in mind Remark~\ref{CP=rectang + afforded}. By the proof of this theorem one also may deduce that viewing $\varphi^{-1}_{\alpha_1,\overline{X_1^{h_1}a}}$ in $\F(r_1)$ and $\varphi^{-1}_{\alpha_2,\overline{X_2^{h_2}b}}$ in $\F(r_2)$ it happens that $C_1=\left(\varphi^{-1}_{\alpha_1,\overline{X_1^{h_1}a}}\right)\subseteq \F(r_1)$ and $C_2=\left(\varphi^{-1}_{\alpha_2,\overline{X_2^{h_2}b}}\right)\subseteq \F(r_2)$. It is also clear that the cyclic codes $C_1$ and $C_2$ verify that their minimum distances  equal their respective maximum BCH bounds, as $a$ and $b$ satisfy the conditions in \cite[Corollary 5]{BBSAMC}. 

 In some sense we may consider $C_1$ and $C_2$ as ``projected codes'', but clearly $C$ is not a product of them under any classical algebraic operation; however, for their defining sets the equality under the product occurs, and all properties of ds-bounds are related, as the following results show.
 
 \begin{lemma} \label{BCH multivariate code}
 Under the same notation from previous paragraphs,  let $C$ be an abelian code in $\F(r_1,r_2)$, with 
$M=M\left(\D_{\overline{\alpha}}(C)\right)$  and suppose $\supp\left(M\right)=\pi_1\left(\supp(M)\right)\times \pi_2\left(\supp(M)\right)$. Consider $D_1=\Z_{r_1}\setminus \pi_1\left(\supp(M)\right)$,  $D_2=\Z_{r_2}\setminus \pi_2\left(\supp(M)\right)$  and let $C_i$ be the 
cyclic code with $\D_{\alpha_i}(C_i)=D_i$, for $i\in \{1,2\}$. Then
 \begin{enumerate}
  \item For any set $\B$ of ds-bounds, $\Delta_{\B,\overline{\alpha}}(C)=\Delta_{\B,\alpha_1}(C_1)\cdot \Delta_{\B,\alpha_2}(C_2)$.
  \item $C$ is a nonzero BCH multivariate code if and only if $C_1$ and $C_2$ are BCH cyclic codes in the classical sense (see \cite{MWS}).
 \end{enumerate}
 
 Moreover, if case (2) holds, with $C_i=\left(\alpha_i,\delta_i,b_i\right)$, for $i\in\{1,2\}$, then $C=B_q\left((\alpha_1,
\alpha_2),\{1,2\},\{\delta_1,\delta_2\},\{b_1,b_2\}\right)$.
 \end{lemma}
\begin{proof}
 Assertion \textit{(1)} comes from Corollary~\ref{propiedades matrices rectangulares}, having in mind Remark~\ref{CP=rectang + afforded}.
 
 Now we prove assertion \textit{(2)}. First, suppose that $C$ is a multivariate BCH code. Assume that $1\in \bar \gamma$, and let $B=\{\overline{b_1},\dots,\overline{b_1+\delta_1-2}\}$ be its list of consecutive integers 
modulo $r_1$. Consider a $q$-cyclotomic coset $T\subseteq D_1$ and take $t\in T$. Since $t\in D_1$ 
then the set $(t,\Z_{r_2})\subseteq \D_{\overline{\alpha}}(C) $ (with the obvious meaning). 
 
 If $(t,\Z_{r_2})\cap Q(\overline{b_1+l},\Z_{r_2})\neq \emptyset$, for some $l\in\{0,\dots,\delta_1-2\}$, then we are done. Otherwise, it must happen  $2\in \bar \gamma$ and $(t,\Z_{r_2})\subset \bigcup_{l=0}^{\delta_2-2}Q(\Z_{r_1},\overline{b_2+l})$. Since $C\neq 0$ then  we may take an element $u\in \Z_{r_2}\setminus \bigcup_{l=0}^{\delta_2-2}C_q(\overline{b_2+l})$ and clearly $(t,u)\not\in \bigcup_{l=0}^{\delta_2-2}Q(\Z_{r_1},\overline{b_2+l})$, a contradiction.
 
 The final assertion comes immediately from the fact that $M$ is a CP-matrix.
\end{proof}

\begin{theorem} \label{construccion bch bi con ciclicos}
Let $\mathbb{K}$ be an intermediate field $\F\subseteq \mathbb{K} \subseteq \Le$, $a=a(X_1)\in \mathbb{K}(r_1)$ and 
$b=b(X_2)\in \mathbb{K}(r_2)$ be such that $a\mid X_1^{r_1}-1$ and $b\mid X_2^{r_2}-1$. If there exist $(\alpha_1,\alpha_2)\in U$, 
$h_1\in \Z_{r_1}$ and $h_2\in\Z_{r_2}$  for which $\varphi^{-1}_{\alpha_1,\overline{X_1^{h_1}a}}\in \F(r_1)$ and 
$\varphi^{-1}_{\alpha_2,\overline{X_2^{h_2}b}}\in \F(r_2)$, with at least one of the inverses different from $1$,  
then there exists in $\F(r_1,r_2)$ a family of permutation equivalent BCH multivariate  codes $\left\{C_{\overline{\beta}}
=B_q(\overline{\beta},\bar \gamma,\bar \delta,\bar b)\tq \overline{\beta}\in U\right\}$   such that 
\begin{enumerate}
 \item $\bar \gamma\subseteq\{1,2\}$  and
  \begin{enumerate}
   \item If $\supp(a)=\Z_{r_1}$ then $1\not \in \bar \gamma$.
   \item If $\supp(b)=\Z_{r_2}$ then $2\not \in \bar \gamma$.
  \end{enumerate}
 \item $\bar \delta=\left\{\delta_k\tq k\in \bar \gamma\right\}$ with $\delta_1=\Delta(M(a))$ and $\delta_2=\Delta(M(b))$.
\item $\displaystyle{\prod_{k\in\bar \gamma}}\delta_k=\Delta(C_{\overline{\beta}})=d(C_{\overline{\beta}})$, for each 
$\bar\beta\in U$.
  \item $\varphi^{-1}_{\beta_1,\overline{X_1^{h_1}a}}\cdot \varphi^{-1}_{\beta_2,\overline{X_2^{h_2}b}}=\varphi^{-1}_{\overline{\beta},\overline{X_1^{h_1}a}\overline{X_2^{h_2}b}}\in C_{\overline{\beta}}$, where $\overline{\beta}=(\beta_1,\beta_2)$.
\end{enumerate}
\end{theorem}
\begin{proof} 
 Set $\bar b=\left\{b_k\tq k\in\bar \gamma\right\}$. First, suppose $\supp(M(a))=\supp(a)\neq \Z_{r_1}$ and $\supp(M(b))=\supp(b)\neq \Z_{r_2}$. Then by \cite[Theorem 2]{BBSAMC}, for each 
$\overline{\beta}\in U$, there exist BCH cyclic codes $B(a)=B_q(\beta_1,\delta_1,b_1)$ and $B(b)=B_q(\beta_2,\delta_2,b_2)$ such that $\delta_1=\Delta(M(a))$, $\delta_2=\Delta(M(b))$, $\varphi^{-1}_{\beta_1,\overline{X_1^{h_1}a}} \in B(a)$ and $ \varphi^{-1}_{\beta_2,\overline{X_2^{h_2}b}}\in B(b)$. Now let $C=C_{\overline{\beta}}$ be the abelian code with  $\overline{\D_{\overline{\beta}}(C)}=\overline{\D_{\beta_1}(B(a))}\times \overline{\D_{\beta_2}(B(b))}$ and set $M=M\left(\D_{\overline{\beta}}(C)\right)$. It is clear that $M$ is a CP-matrix and, moreover, following the notation of Lemma~\ref{BCH multivariate code}, $D_1=\D_{\beta_1}(B(a))$ and $D_2=\D_{\beta_2}(B(b))$.  Then $C$ is a bivariate BCH code with $\bar \gamma=\{1,2\}$, $\bar \delta=\{\delta_1,\delta_2\}$ and $b=\{b_1,b_2\}$. Now by statement \textit{(1)} of Lemma~\ref{BCH multivariate code}, we have $\Delta_{\overline{\beta}}(C)=\delta_1\delta_2$ and clearly statement \textit{ (5)} of this theorem holds.
  
  It remains to prove the equality $\Delta(C_{\overline{\beta}})=d(C_{\overline{\beta}})$. On the one hand, we have $\Delta_{\overline{\beta}}(C)=\Delta(M(a))\Delta(M(b))=\Delta(M(ab))$ and, on the other hand, by hypothesis,   $\Delta(M(ab))=\left|\overline{Z(ab)}\right|$. Hence, by applying Theorem~\ref{caracterizacion codigos multi distancia real},  we are done.
\end{proof}

We may take, again, advantage from cyclic codes to transform a given abelian code $C=(g)$, with $d(C)=\Delta(C)$ into another abelian code with higher dimension, as in Example~\ref{ejemplo mejorar la dimension 3 por 45}, until to get a new BCH code.

\begin{example}\rm{
 We continue with the code $C$ from Example~\ref{ejemplo 3 por 45}. Recall that $q=2$, $r_1=3$, $r_2=45$ and we have fixed $\alpha_1\in U_3$ and $\alpha_2 \in U_{45}$. We have polynomials $a=X+1$ and $b=Y^{40}+Y^{39}+Y^{38}+ Y^{36}+Y^{35}+Y^{32}+Y^{30}+Y^{25}+Y^{24}+Y^{23}+Y^{21}+ Y^{20}+Y^{17}+ Y^{15}+Y^{10}+Y^9+ Y^8+Y^6+Y^5+Y^2+1$ such that $a\mid X_1^3-1$ and $h_1=1$ works, and $b\mid X_2^{45}-1$ and $h_2=5$ works in the sense of Theorem~\ref{construccion bch bi con ciclicos}. Hence the hypothesis of this theorem are satisfied.
 
 Now we have to follow the proof to construct our multivariate BCH code. The proof of Theorem~\ref{construccion bch bi con ciclicos}  uses a construction from \cite[Theorem 2]{BBSAMC}. Clearly $B(a)$ is the BCH code in $\mathbb{F}_2(3)$ with defining set $D_{\alpha_1}(B(a))=C_2(1)$. On the other hand, by \cite[Example 8]{BBSAMC}, the code $B(b)$ has defining set $D_{\alpha_2}(B(b))=C_2(1)\cup C_2(3)$ so that it is a BCH code. In fact, $B(b)=B_2(\alpha_2,5,1)$, following the usual notation for BCH codes.
 
 Thus $C_{(\alpha_1,\alpha_2)}=B_2\left((\alpha_1,\alpha_2),\{1,2\},\{2,5\},\{1,1\}\right)$, $d\left(C_{(\alpha_1,\alpha_2)}\right)=10$ and $\dim_{\mathbb{F}_2}\left(C_{(\alpha_1,\alpha_2)}\right)=58$. This code has better parameters than the code $C$ from Example~\ref{ejemplo 3 por 45} and  $C'$  from Example~\ref{ejemplo mejorar la dimension 3 por 45}.
}\end{example}

We finish by extending Corollary~\ref{construccion distancia verdadera con  irreducibles} to bivariate BCH codes.

\begin{corollary}
 Let $\mathbb{K}$ be an intermediate field $\F\subseteq \mathbb{K} \subseteq \Le$ and $a=a(X_1)\in \Le(r_1)$ be 
such that $a\mid X_1^{r_1}-1$, with $\varphi^{-1}_{\alpha_1,\overline{X_1^{h_1}a}}\in \F(r_1)$, for some $\alpha_1\in U_{r_1}$ 
and $h_1\in \Z_{r_1}$. 
 
 Let  $g$ be an irreducible factor of $X_2^{r_2}-1$ in $\mathbb{K}[X_2]$ with defining set $D_{\alpha_2}(g)$, 
for some $\alpha_2\in U_{r_2}$. Set $b=(X_2^n-1)/h$.  If there are positive integers $j,t$ such that 
$b(\alpha_2^j)=\alpha_2^t$ and $\gcd \left(j,\frac{r_2}{\gcd(q-1,r_2)}\right)\mid t$ then
there exists a bivariate BCH code  $C=B_q(\overline{\alpha},\bar \gamma,\bar \delta,\bar b)$ 
in $\F(r_1,r_2)$ verifying $\Delta\left(M(ab)\right)=\Delta(C)=d(C)$, for certain $\overline{\alpha},
\bar \gamma,\bar \delta,\bar b$.
\end{corollary}

\begin{proof}
 Comes immediately from Corollary~\ref{construccion distancia verdadera con  irreducibles} together with Theorem~\ref{construccion bch bi con ciclicos}.
\end{proof}

\begin{example}\rm{
 We shall extend to bivariate BCH codes those abelian codes on Table~\ref{tabela1} from 
Example~\ref{extender codigos a partir de minimales}. Recall that we had a list of divisors $a_i$ of $X_1^{7}-1$ in 
$\mathbb{F}_2[X_1]$ and divisors $b_j$ of $X_2^{15}-1$ in $\mathbb{F}_2[X_2]$, namely, $a_1=1+X_1$, $a_2= 1+X_1+X_1^3$, 
$a_2=1+X_1^2+X_1^3$, $b_1=\frac{X_2^{15}-1}{1+X_2+X_2^2}$, $b_2=\frac{X_2^{15}-1}{1+X_2+X_2^4}$ and 
$b_3=\frac{X_2^{15}-1}{1+X_2^3+X_2^4}$ from which we constructed the mentioned table. 
 
 Now, one may check easily that the codes determined by $a_2$, $a_3$ and $a_1a_3$ are all BCH. Specifically, 
$B(a_2)=B_2\left(\alpha_1,4,5\right)$, $B(a_3)=B_2\left(\alpha_1,4,0\right)$ and $B(a_1a_3)=B_2\left(\alpha_1,3,5\right)$, while the code determined by  $a_2a_3$ is all $\mathbb{F}_2(7)$.
On the other hand, as it is shown in \cite[Example 9]{BBSAMC} one may construct from $b_1$ the code $B(b_1)=B_2(\alpha_1,2,0)$, of dimension $14$, from $b_2$ the code $B_2(\alpha_1,4,13)$ of dimension 10, and from $b_3$  the code $B_2(\alpha_1,4,0)$ which also has dimension 10.
 
 Thus, Table~\ref{tabela3} is the new table of bivariate BCH codes in $\mathbb{F}_2(7,15)$:

 \begin{table}[h]
\[\begin{array}{|l|l|l|l|} \hline 
\bar \gamma& \bar{b} &  \text{Dimension} & \Delta=d \\  \hline 
 \{1,2\}& \{5,0\}    & 42  & 8 \\  \hline 
 \{1,2\} & \{5,13\}    & 40 &  16\\  \hline 
 \{1,2\} & \{5,0\}   & 40 & 16 \\  \hline 
  \{1,2\} & \{0,0\}   & 42 & 8 \\  \hline 
\{1,2\} & \{0,13\}   & 40 &  16  \\  \hline 
 \{1,2\} & \{0,0\}   & 40 &  16  \\  \hline 
  \{1,2\} & \{5,0\}  & 56 & 6 \\  \hline 
\{1,2\}  & \{5,13\}   & 40 & 12 \\  \hline 
 \{1,2\} & \{5,0\}   & 40 & 12 \\  \hline 
   \{2\} & \{0\} & 98 & 2 \\  \hline 
\{2\}  & \{13\}   & 70 & 4 \\  \hline 
 \{2\} & \{0\}   & 70 & 4 \\  \hline 
\end{array}\]
\caption{Bivariate BCH codes in $\mathbb{F}_2(7,15)$:\label{tabela3}}
\end{table}
 \vspace{1cm}

In the case of codes in $\mathbb{F}_2(5,21)$, the code determined by $a=\Phi_5$ is $\mathbb{F}_2(5)$, so  we construct $B(b'_1)=B_2\left(\alpha_2,2,0\right)$, $B(b'_2)=B_2\left(\alpha_2,3,19\right)$, $B(b'_3)=B_2\left(\alpha_2,3,1\right)$, $B(b'_4)=B_2\left(\alpha_2,6,17\right)$ and $B(b'_5)=B_2\left(\alpha_2,6,0\right)$. Table~\ref{tabela4} is the new table of bivariate BCH codes in $\mathbb{F}_2(5,21)$.
 
 \begin{table}[h]
\[\begin{array}{|l|l|l|l|} \hline 
\bar \gamma& \bar{b} &  \text{Dimension} & \Delta=d \\  \hline 
 \{2\}& \{0\}    & 100  & 2 \\  \hline 
 \{2\} & \{19\}    & 75 &  3\\  \hline 
 \{2\} & \{1\}   & 75 & 3 \\  \hline 
  \{2\} & \{17\}   & 55 & 6 \\  \hline 
\{2\} & \{0\}   & 55 &  6  \\  \hline 
\end{array}\]
\caption{Bivariate BCH codes in $\mathbb{F}_2(5,21)$.\label{tabela4}}
\end{table}
}\end{example}

\section{Conclusion}
We have developed a technique to extend any bound for the minimum distance of cyclic codes constructed from its defining sets (ds-bounds) to abelian (or multivariate) codes through the notion of $\B$-apparent distance. We used this technique to improve the searching for new bounds for abelian codes having unknown minimum distance. We have also studied conditions for an abelian code to verify that its  $\B$-apparent distance reaches its (true) minimum distance and we have constructed some tables of such codes as an application.

%





\begin{thebibliography}{1}

\bibitem{BJ} {E.R. Berlekamp and J. Justesen}, \ \textit{Some long cyclic linear binary codes are not so bad}. IEEE Transactions on Information Theory \textbf{20}   (1974), 351-356.

\bibitem{BSpermdec} {J.J. Bernal, J.J. Sim\'on}, \ \textit{Partial permutation decoding for abelian codes}. IEEE Transactions on Information Theory \textbf{59} (8) (2013), 5152-5170.

\bibitem{BBCS2} {J.J. Bernal, D.H. Bueno-Carre\~no, J.J. Sim\'on}, \ \textit{Apparent distance and a notion of BCH multivariate codes}. IEEE Transactions on Information Theory, \textbf{62}(2) (2016), 655-668.

\bibitem{BGSISITA} {J.J. Bernal, M. Guerreiro, J.J. Sim\'on}, \ \textit{Ds-bounds for cyclic codes: new bounds for abelian codes}. Proceedings of ISITA 2016, Monterey, CA, USA, 712-716.

\bibitem{BettiSala} {E. Betti, M. Sala}, \ \textit{A new bound for the minimum distance of a cyclic code from its defining set}. IEEE Transactions on Information Theory \textbf{52} (8) (2006) 3700-3706.

\bibitem{Blah} {R.E. Blahut}, \ \textit{Decoding of cyclic codes and codes on curves}. In W.C. Huffman and V. Pless (Eds.), 
\emph{Handbook of Coding Theory}.  Vol. II, 1569-1633, 1998. 

\bibitem{BVtabla} A. E. Brouwer, T. Verhoeff, \ \textit{An updated table of the minimum-distance bounds for binary linear codes}.  IEEE Transactions on Information Theory \textbf{39} (2) (1993) 662--677.

\bibitem{BBS1var}  D. H. Bueno-Carre\~no, J.J. Bernal and J.J. Sim\'{o}n, \ \textit{A characterization of cyclic codes whose minimum distance equals their maximum BCH bound}. Proceedings ACA 2013, M\'{a}laga, 109-113.

\bibitem{BBSAMC}  D. H. Bueno-Carre\~no, J.J. Bernal and J.J. Sim\'{o}n, \textit{Cyclic and BCH codes whose minimum distance equals their maximum BCH bound}. Advances in Mathematics of Communications 10 (2016), 459-474. 


\bibitem{Camion} {P. Camion}, \ \textit{Abelian Codes}. MCR Tech. Sum. Rep. 1059, University of Wisconsin, Madison, 1970.

\bibitem{Gtabla}  M. Grassl, \ \textit{Bounds on the minimum distance of linear codes and quantum codes.} Online available 
at http://www.codetables.de.

\bibitem{HTY}  {C.R.P. Hartmann,  H. Tai-Yang}, \ \textit{Some results on the weight structure of cyclic codes of composite lenght}. IEEE Transactions on Information Theory \textbf{22} (3) (1976) 340-348.

\bibitem{HT}  {C.R.P. Hartmann,  K.K. Tzeng}, \ \textit{Generalizations of the BCH bound}. Information and Control   
\textbf{20} (1972) 489-498.

\bibitem{Imai} {H. Imai}, \ \textit{A theory of two-dimensional cyclic codes}. Information and Control \textbf{34} (1)  (1977) 1-21.

\bibitem{Jensen} {J.M. Jensen}, \ \textit{The concatenated structure of cyclic and abelian codes}. IEEE Transactions on Information Theory \textbf{31} (6)  (1985) 788-793.

\bibitem{KZ} {T. Kaida, J. Zhen}, \ \textit{A decoding method up to the Hartmann-Tzeng bound using the DFT for cyclic codes}. In Proceedings of Asia-Pacific Conference on Communications  (2007) 409-412.
 

\bibitem{MWS} {F.J. MacWilliams and N.J.A. Sloane}, \ \textit{The Theory of Error-Correcting Codes}. Elsevier, Amsterdam, 1977.

\bibitem{Roos1} {C. Roos},\ \textit{A new lower bound for the minimum distance of a cyclic codes}.  IEEE Transactions on Information Theory \textbf{29} (3)  (1983) 330-332.

\bibitem{Roos} {C. Roos},\ \textit{A generalization of the BCH bound for cyclic codes, including the Hartmann-Tzeng bound}. J. Combinatorial Theory, Series A  (1982) 229-232.

\bibitem{Evans} {R.E. Sabin},\ \textit{On minimum distance bounds for abelian codes}. Applicable Algebra Eng. Commun. Comput. \textbf{3} (3) (1992) 183-197.

\bibitem{Sakata} {S. Sakata},\ \textit{Decoding binary cyclic 2-D codes by the 2-D Berlekamp-Massey algorithm}. IEEE Transactions on Information Theory \textbf{37} (4)  (1991) 1200-1203.

\bibitem{vLW} {J.H. van Lint and R.M. Wilson}, \ \textit{On the minimum distance of cyclic codes}. IEEE Transactions on Information Theory \textbf{32} (1)  (1986) 23-40.

\bibitem{Yamada} {K.-C. Yamada}, \ \textit{Hypermatrix and its applications}. Hitotsubashi J. Arts Sciences 
\textbf{6}(1)  (1965) 34-44.

\bibitem{ZJtabla} A. Zeh, T. Jerkovits. Cyclic Codes Online available at http://www.boundtables.org

\bibitem{ZK1} {J. Zhen, T. Kaida}, \ \textit{The designed minimum distance of medium length for binary cyclic codes}. ISITA-2012, Honolulu, Hawaii (2012) 441-445.
\end{thebibliography}
\end{document}